\documentclass[11pt, DIV10,a4paper]{article}
\usepackage{color}
\usepackage{float}
\usepackage[bf,small]{caption2}
\usepackage{comment}
\usepackage{amsmath}
\usepackage{bigints}
\usepackage{rotating, booktabs}
\usepackage{amsopn}
\usepackage{amsfonts}
\usepackage{amsthm}
\usepackage{amssymb}
\usepackage{amsbsy}
\usepackage{multirow}
\usepackage{slashbox}
\usepackage{natbib}
\usepackage{tocbibind}
\usepackage[a4paper,colorlinks,breaklinks,bookmarksopen,bookmarksnumbered]{hyperref}
\usepackage[refpage]{nomencl}
\usepackage{makeidx}
\usepackage{pst-all}
\usepackage{epsfig,psfrag}
\usepackage{graphicx}
\usepackage{amsmath}
\usepackage{epstopdf}
\usepackage{tabularx}
\usepackage{subfigure}
\usepackage{setspace}
\usepackage[mathscr]{euscript}
\usepackage[margin=1in]{geometry}
\newcommand{\bzero}{\boldsymbol{0}}
\newcommand{\bone}{\boldsymbol{1}}
\newcommand{\bmu}{\boldsymbol{\mu}}
\newcommand{\bSigma}{\boldsymbol{\Sigma}}
\newcommand{\bI}{\bf{I}}

\newcommand {\ctn}{\cite}

\usepackage{url}
\newtheorem{lemma}{Lemma}[section]

\newtheorem{theorem}{Theorem}[section]

\normalsize

\begin{document}
\title{\textbf{On Geometric Ergodicity of Additive and Multiplicative Transformation Based Markov Chain Monte Carlo
in High Dimensions}}
\author{ Kushal Kr. Dey$^{\dag}$ , Sourabh Bhattacharya$^{\ddag, +}$ }
\date{}
\maketitle
\begin{center}
$^{\dag}$  University of Chicago \\
$^{\ddag}$ Indian Statistical Institute\\
$+$ Corresponding author:  \href{mailto: bhsourabh@gmail.com}{bhsourabh@gmail.com}
\end{center}

\begin{abstract}
Recently \ctn{Dutta13} introduced a novel Markov Chain Monte Carlo methodology that can simultaneously update 
all the components of high dimensional parameters using simple deterministic transformations of 
a one-dimensional random variable drawn from any arbitrary distribution
defined on a relevant support. The methodology, which the authors refer to as Transformation-based
Markov Chain Monte Carlo (TMCMC), greatly enhances computational speed and acceptance rate
in high-dimensional problems. 
Two significant transformations associated with TMCMC are additive and multiplicative transformations. 
Combinations of additive and multiplicative transformations are also of much interest.
In this work we investigate geometric ergodicity associated with additive and multiplicative TMCMC,
along with their combinations, assuming that the target distribution is multi-dimensional and belongs to the
super-exponential family; we also illustrate their efficiency in practice with simulation
studies.
\\[2mm]
{\bf Keywords}: {\it Acceptance Rate; Geometric Ergodicity; High Dimension; Mixture; Proposal Distribution; 
Transformation-based Markov Chain Monte Carlo.}


\end{abstract}

\section{Introduction}

It is well-known that in high dimensions traditional Markov Chain Monte Carlo (MCMC) methods,
such as the Metropolis-Hastings algorithm,
face several challenges, with respect to computational complexity, as well as with 
convergence issues. Indeed, Bayesian computation often requires inversion of high-dimensional matrices
in each MCMC iteration, causing enormous computational burden. Moreover, such high-dimensional problems
may converge at an extremely slow rate, because of the complicated posterior dependence among the parameters.
This implies the requirement of an extremely large number of iterations, but since even individual iterations
may be computationally burdensome, traditional MCMC methods do not seem to be ideally suited for Bayesian analysis
of complex, high-dimensional problems.

In an effort to combat the problems \ctn{Dutta13} proposed a novel methodology
that can update all the parameters simultaneously in a single block using 
simple deterministic bijective transformations of 
a one-dimensional random variable (or any other low-dimensional random variables) 
drawn from some arbitrary distribution.
The idea effectively reduces the high-dimensional random parameter to a one-dimensional parameter,
thus dramatically improving computational speed and acceptance rate. Details are provided in \ctn{Dutta13}.

Among the deterministic, bijective transformations, \ctn{Dutta13} recommend the additive and the multiplicative
transformations. Here it is important to mention that the multiplicative transformation is designed to update
parameters on the real line, not just on $(0,\infty)$, and thus, can not be represented as 
the log-additive transformation.
In Sections \ref{subsec:additive_tmcmc} and \ref{subsec:multiplicative_tmcmc} we provide brief overviews 
of additive and multiplicative TMCMC, respectively. In Section \ref{subsec:add_mult_tmcmc} we briefly
explain additive-multiplicative TMCMC, which is a combination of additive and multiplicative TMCMC.

This paper deals with geometric ergodicity (or geometric rate of convergence) of the 
TMCMC chain (both additive and multiplicative,
along with their mixtures of two kinds) 
to the multi-dimensional stationary distribution. 
The geometric ergodicity property, apart from theoretically ensuring convergence of the underlying Markov chain
to the stationary distribution at a geometric rate,
also ensures asymptotic stability of a regular 
family of stochastic estimates through 
the application of the central limit theorem (see \ctn{Meyn93}, Chapter 17, and \ctn{jones01}, Section 5.3). 
The geometric ergodicity of the Random Walk Metropolis Hastings (RWMH) chain is already well documented 
(see \ctn{Mengersen96}, 
\ctn{Roberts96}, \ctn{Jarner00}). Some extensions of these 
results to chains with polynomial rates of convergence and specific forms of target densities 
(for instance, heavy tailed families) are also available in the literature (\ctn{Jarner02}, \ctn{Jarner07}). 
In this paper we present conditions that guarantee geometric ergodicity of the TMCMC chain 
corresponding to both additive and multiplicative moves, when the target distribution is multi-dimensional. 
Crucially, we assume that the target distribution belongs to the super-exponential family. Note that
the super-exponential assumption has also been crucially used by \ctn{Jarner00} for proving geometric ergodicity
of RWMH for multi-dimensional target distributions.

While dealing with multiplicative TMCMC,
we encounter a technical problem, which is bypassed by forming an appropriate mixture of 
additive and multiplicative moves, to which we refer as ``essentially fully" multiplicative TMCMC. 
We also consider a usual mixture of additive and multiplicative moves. We establish 
geometric ergodicity of both kinds of mixtures and demonstrate with simulation studies
that the usual mixture outperforms RWMH, additive TMCMC, as well as ``essentially fully" multiplicative TMCMC.

In Section \ref{sec:geo_additive}, we give conditions 
for geometric ergodicity of additive TMCMC. The approach to establishing geometric ergodicity of the TMCMC chains 
associated with multiplicative TMCMC is more complicated and is covered in detail in 
Section \ref{sec:geo_multiplicative}. In 
Section \ref{sec:simulation}, we illustrate the practical implications of our theoretical results by 
conducting simulation studies, where we numerically compare convergence issues 
of the TMCMC approach with that of RWMH, especially in high dimensions. 
In Section \ref{sec:non_super_exponential}, we discuss extension of our approach to situations where
the high-dimensional target densities are not in the super-exponential family but can be dealt with
using special techniques, in particular, a diffeomorphism based method developed by \ctn{Johnson12}, and conduct detailed
simulation studies in such set-up, demonstrating that TMCMC very significantly outperforms RWMH
that set-up.
Concluding remarks are provided in Section \ref{sec:conclusions}.

\subsection{Additive TMCMC}
\label{subsec:additive_tmcmc}

Suppose that we are simulating from a $d$ dimensional space (usually $\mathbb{R}^{d}$), and suppose we are currently at a point 
$x= (x_{1}, \ldots, x_{d})$.
Let us define $d$ random variables $b_{1}, \ldots, b_{d}$, such that, for $i=1,\ldots,d$, 
\begin{equation}
b_{i} =\left\{\begin{array}{ccc} +1 & \mbox{with probability} & p_i; \\
 -1 & \mbox{with probability} & 1-p_i.
 \end{array}\right.
 \label{eq:b_add}
\end{equation}
The additive TMCMC uses moves of the following type: 
\begin{equation*}
(x_{1}, \ldots, x_{d}) \rightarrow (x_{1}+ b_{1}\epsilon, \ldots, x_{d}+b_{d}\epsilon), 
\end{equation*}
where $\epsilon\sim g^{(1)}=q^{(1)}(\cdot)I_{\{\epsilon>0\}}$. Here $q^{(1)}(\cdot)$ is an arbitrary density with support $\mathbb R_+$, the
positive part of the real line, and
for any set $A$, $I_{A}$ denotes the indicator function of $A$. We define $T^{(1)}_b(x,\epsilon)=(x_1+b_1\epsilon,\ldots,x_d+b_d\epsilon)$ 
to be the additive transformation of $x$ corresponding to the `move-type' $b$.
In this work,
we shall assume that $p_i=1/2$ for $i=1,\ldots,d$.
Note that the Jacobian of the additive transformations is one. 

Thus, a single $\epsilon$ is simulated from $q(\cdot)I_{\{\epsilon>0\}}$, which is then either added to, or subtracted
from each of the $d$ coordinates of $x$ with probability $1/2$. Assuming that the target distribution is proportional to 
$\pi$, the new move $T^{(1)}_b(x,\epsilon)$, corresponding to the move-type $b$,
is accepted with probability
\begin{equation}
\alpha=\min\left\{1,\frac{\pi(T^{(1)}_b(x,\epsilon))}{\pi(x)}\right\}.
\label{eq:acc_additive}
\end{equation}

The path diagram for additive TMCMC that displays the possible regions to which 
our chain can move to starting from a fixed point, is presented in Figure~\ref{fig:addTMCMC}.
\begin{figure}
\centering
\includegraphics[width=7cm,height=7cm]{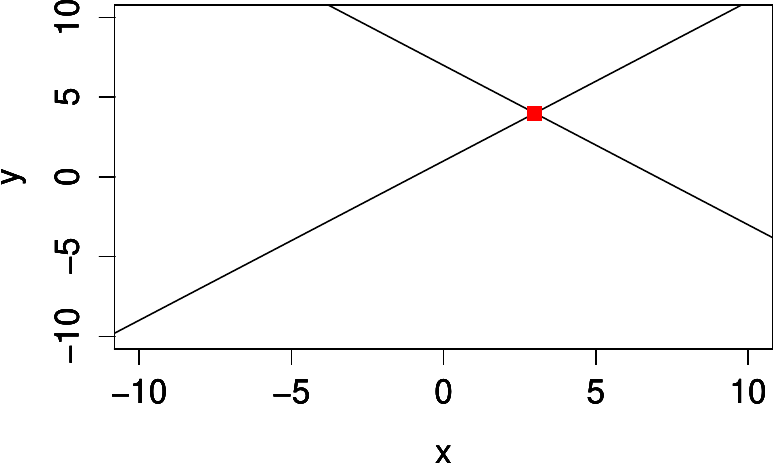}
\caption{Path diagram for Additive TMCMC in one step from a fixed point denoted by the red patch in the
middle.}
\label{fig:addTMCMC}
\centering
\includegraphics[width=7cm,height=7cm]{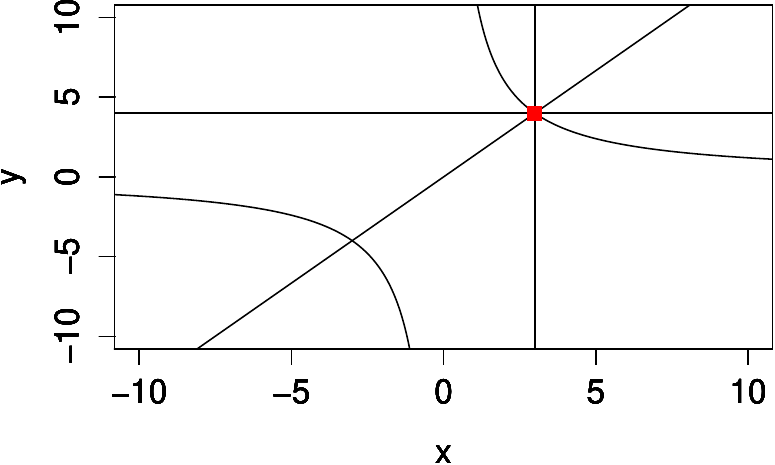}
\caption{Path diagram for the first version of Multiplicative TMCMC in one step from a 
fixed point denoted by the red patch.}
\label{fig:fullMultiplicative}
\centering
\includegraphics[width=7cm,height=7cm]{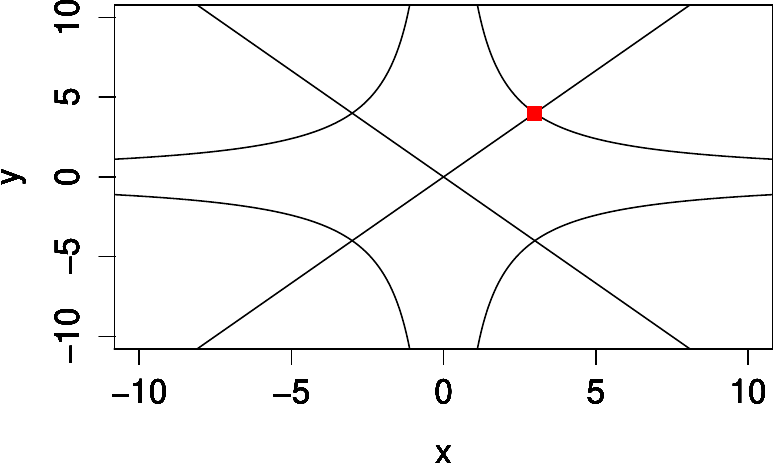}
\caption{Path diagram for the second version of Multiplicative TMCMC in one step 
from a fixed point denoted by the red patch.}
\label{fig:fullMultiplicative2}
\end{figure}

In this paper we show, under appropriate and reasonably general assumptions on $\pi$, that additive 
TMCMC with $p_i=1/2;~i=1,\ldots,d$, is geometrically ergodic for any finite dimension $d$.

\subsubsection{Discussion on non-uniform move-type probabilities for additive TMCMC}
\label{subsubsec:additive_non_uniform_move_type}
For simplicity of illustration, let us assume that $p_i=p;~i=1,\ldots,d$. Also, let 
$Y=\sum_{i:b_i=1}b_i$. Then $Y\sim \mbox{Binomial}\left(d,p\right)$.
The acceptance probability is then given by
\begin{align}
\alpha&=\min\left\{1,\left(\frac{p}{1-p}\right)^Y\left(\frac{1-p}{p}\right)^{(d-Y)}
\frac{\pi(T^{(1)}_b(x,\epsilon))}{\pi(x)}\right\}\notag\\
&=\min\left\{1,\left(\frac{p}{1-p}\right)^{(2Y-d)}
 \frac{\pi(T^{(1)}_b(x,\epsilon))}{\pi(x)}\right\}
\label{eq:acc_additive2}
\end{align}
Now, as $d\rightarrow\infty$, $(2Y-d)\stackrel{a.s}{\sim} (2dp-d)=d(2p-1)$, where, for any 
two random sequences $\left\{m_d;~d=1,2,\ldots\right\}$ and 
$\left\{n_d;~d=1,2,\ldots\right\}$, $m_d\stackrel{a.s.}{\sim} n_d$ 
indicates $\underset{d\rightarrow\infty}{\lim}\frac{m_d}{n_d}=1$, almost surely.
Hence, for $p\neq 1/2$, $\left(\frac{p}{1-p}\right)^{(2Y-d)}\stackrel{a.s.}{\rightarrow}\infty$.

Now note that, for additive TMCMC with a single $\epsilon$, as $d\rightarrow\infty$, 
the ratio $\frac{\pi(T^{(1)}_b(x,\epsilon))}{\pi(x)}$
is expected to converge to zero at a very slow rate. Indeed, it follows from the supplement of \ctn{Dutta13} 
that under the strong log-concavity assumption on $\pi$, 
the acceptance rate with these non-uniform move-type probabilities satisfies the following inequalities
as $d\rightarrow\infty$:
\begin{equation}
\left\{2\Phi\left(\sqrt{-\frac{2}{dM_d}\log\frac{1-\psi_2}{c_d}}\right)-1\right\}\leq 
AR_p\leq \left\{2\Phi\left(\sqrt{-\frac{2}{dM_d}\log\frac{\psi_1}{c_d}}\right)-1\right\},
\label{eq:AR_non_uniform}
\end{equation}
where $0<\psi_1,\psi_2<1$, $M_d=O\left(d^t\right);~t>2$, and $c_d = \left(\frac{p}{1-p}\right)^{d(2p-1)}$. 
For $p=1/2$, we obtain the following asymptotic inequality proved in the supplement of \ctn{Dutta13} 
\begin{equation}
\left\{2\Phi\left(\sqrt{-\frac{2}{dM_d}\log (1-\psi_2)}\right)-1\right\}\leq 
AR_{\frac{1}{2}}\leq \left\{2\Phi\left(\sqrt{-\frac{2}{dM_d}\log \psi_1}\right)-1\right\},
\label{eq:AR_uniform}
\end{equation}
which is a special case of (\ref{eq:AR_non_uniform}).
It is shown in the supplement of \ctn{Dutta13} that for $p=1/2$,
as $d\rightarrow\infty$, the acceptance rate of additive TMCMC tends to zero at a much slower rate compared to 
that of the normal random walk Metropolis-Hastings algorithm.
In fact, it is easy to see that $AR_p\rightarrow 0$ as $d\rightarrow\infty$ for any $p\in (0,1)$, and
quite importantly, it holds that
$\frac{AR_p}{AR_{\frac{1}{2}}}\rightarrow\infty$ as $d\rightarrow\infty$.
In other words, for high-dimensional target distributions, 
the additive TMCMC based acceptance rate can be further improved with non-uniform
move probabilities. 

But an increase in acceptance rate does not necessarily lead to faster convergence of the underlying 
Markov chain. Hence, although higher acceptance rates are to be expected of additive TMCMC for non-uniform
move-type probabilities in high dimensions, faster rates of convergence may not still be achieved. 
We reserve the investigation of the effects of non-uniform move-type probabilities on convergence rate
for our future research.

\subsection{Multiplicative TMCMC}
\label{subsec:multiplicative_tmcmc}

Again suppose that we are simulating from a $d$ dimensional space (say, $\mathbb R^d$), and that we are currently at a point 
$x= (x_{1}, \ldots, x_{d})$.
Let us now modify the definition of the random variables $b_{1}, \ldots, b_{d}$, such that, for $i=1,\ldots,d$, 
\begin{equation}
b_{i} =\left\{\begin{array}{ccc} +1 & \mbox{with probability} & p_i; \\
0 & \mbox{with probability} & q_i;\\
 -1 & \mbox{with probability} & 1-p_i-q_i.
 \end{array}\right.
 \label{eq:b_mult}
\end{equation}

Let $\epsilon\sim g^{(2)}=q^{(2)}(\cdot)I_{\{|\epsilon|\leq 1\}}$. If $b_i=+1$, then $x_i\rightarrow x_i\epsilon$, if $b_i=-1$,
then $x_i\rightarrow x_i/\epsilon$ and if $b_i=0$, then $x_i\rightarrow x_i$, that is, $x_i$ remains unchanged.
Let the transformed coordinate be denoted by $x^*_i$.
Also, let $J(b,\epsilon)$ denote the Jacobian of the transformation $(x,\epsilon)\mapsto (x^*,\epsilon)$. 
We denote $x^*$ by $T^{(2)}_b(x,\epsilon)$, the multiplicative transformation 
$(x,\epsilon)\mapsto (x^*,\epsilon)$ associated with the move-type $b$. 

For example, if $d=2$, then for $b=(1,1)$, $T^{(2)}_b(x,\epsilon)=(x_1\epsilon,x_2\epsilon)$ and the Jacobian is $\epsilon^2$, 
for $b=(-1,-1)$, $T^{(2)}_b(x,\epsilon)=(x_1/\epsilon,x_2/\epsilon)$ and $|J(b,\epsilon)|=\epsilon^{-2}$. 
For $b=(1,-1)$, $b=(-1,1)$, and $b=(0,0)$, $T^{(2)}_b(x,\epsilon)=(x_1\epsilon,x_2/\epsilon)$, $(x_1/\epsilon,x_2\epsilon)$, 
and $(x_1,x_2)$, respectively, and in all these three instances, 
$|J(b,\epsilon)|=1$. 
For $b=(1,0)$ and $b=(0,1)$, 
$T^{(2)}_b(x,\epsilon)=(x_1\epsilon,x_2)$ and $T^{(2)}_b(x,\epsilon)=(x_1,x_2\epsilon)$, respectively, and in both these cases
$|J(b,\epsilon)|=|\epsilon|$. For $b=(-1,0)$ or $b=(0,-1)$, $T^{(2)}_b(x,\epsilon)=(x_1/\epsilon,x_2)$ and
$(x_1,x_2/\epsilon)$, respectively, and the Jacobian is $|\epsilon|^{-1}$ in both these cases.

In general, the transformation of the $i$-th coordinate is given by
$x_i\epsilon^{b_i}$ and the Jacobian is given by $|\epsilon|^{\sum_{i=1}^db_i}$.

The path diagram for multiplicative TMCMC that displays the possible range of 
values that our chain can move to starting from a fixed point is presented in Figure~\ref{fig:fullMultiplicative}.

We envisage another version of multiplicative TMCMC where we first generate 
$\epsilon\sim g^{(3)}=q^{(3)}(\cdot)I_{\{0<\epsilon\leq 1\}}$, and then make the transformation
$x_i\rightarrow c_ix_i\epsilon^{b_i}$, where $c_i$ takes the values $1$ and $-1$ with probabilities
$r_i$ and $1-r_i$, respectively, where $0<r_i<1$. In this case it is permissible to set $q_i$, 
the probability of $b_i=0$, to zero. Observe that the Jacobian remains the same
as in the first version of multiplicative TMCMC.
The path diagram for this version is shown in Figure~\ref{fig:fullMultiplicative2}.


Since the theory of geometric ergodicity remains essentially the same for both versions of multiplicative TMCMC,
we consider only the first version for the theoretical treatment. 
For our purpose, we assume that $p_i=q_i=1/3;~i=1,\ldots,d$.
Then assuming that the target distribution is proportional to 
$\pi$, the new move $T^{(2)}_b(x,\epsilon)$ is accepted
with probability
\begin{equation}
\alpha=\min\left\{1,\frac{\pi(T^{(2)}_b(x,\epsilon))}{\pi(x)}|J(b,\epsilon)|\right\}.
\label{eq:acc_multiplicative}
\end{equation}

\subsubsection{Discussion on non-uniform move-type probabilities for multiplicative TMCMC}
\label{subsubsec:multiplicative_non_uniform_move_type}

For simplicity of illustration, let us assume that $p_i=p$ and $q_i=q$, for $i=1,\ldots,d$. Also, let 
$Y=\sum_{i:b_i=1}b_i$, and $Z=\sum_{i:b_i=0}b_i$. Then $Y\sim \mbox{Binomial}\left(d,p\right)$,
and $Z\sim \mbox{Binomial}\left(d,q\right)$.
The acceptance probability is then given by
\begin{align}
\alpha&=\min\left\{1,\left(\frac{p}{1-p-q}\right)^Y\left(\frac{1-p-q}{p}\right)^{(d-Y-Z)}
\frac{\pi(T^{(1)}_b(x,\epsilon))}{\pi(x)}\right\}\notag\\
&=\min\left\{1,\left(\frac{p}{1-p-q}\right)^{(2Y+Z-d)}
 \frac{\pi(T^{(1)}_b(x,\epsilon))}{\pi(x)}\right\}
\label{eq:acc_multiplicative2}
\end{align}
As $d\rightarrow\infty$, $(2Y+Z-d)\stackrel{a.s}{\sim} d(2p+q-1)$.
If $2p+q>1$, then $p>(1-q)/2$, so that $1-p-q<(1-q)/2$. Hence, $p/(1-p-q)>1$, implying that
$\left(\frac{p}{1-p-q}\right)^{(2Y+Z-d)}\stackrel{a.s.}{\rightarrow}\infty$. If, on the other hand
$2p+q<1$, then $p<(1-q)/2$, and $1-p-q>(1-q)/2$, implying $p/(1-p-q)<1$. Again, this implies
$\left(\frac{p}{1-p-q}\right)^{(2Y+Z-d)}\stackrel{a.s.}{\rightarrow}\infty$. 
In contrast with additive TMCMC, for multiplicative TMCMC the asymptotic form of the acceptance rate is not available yet,
even for strongly log-concave target distributions. However, it is clear that for $2p+q\neq 1$,
the acceptance rate is asymptotically much higher than that associated with $2p+q=1$.

Although for convenience of presentation we prove geometric ergodicity of multiplicative TMCMC assuming 
$p_i=p=1/3$ and $q_i=q=1/3$;$~i=1,\ldots,d$, the steps of our proofs remain
the same for any other $0<p,q<1$, satisfying $2p+q=1$ (note that for such $p,q$, $p+q=(1+q)/2<1$ is automatically satisfied).
We reserve the cases $2p+q\neq 1$ for our future investigation.

Apart from additive and multiplicative TMCMC we also consider appropriate geometric ergodic mixtures 
of additive and multiplicative TMCMC, which not only help 
bypass a somewhat undesirable theoretical assumption regarding the high-dimensional target density $\pi$, 
but as simulation studies demonstrate, appropriate mixtures of additive and multiplicative 
TMCMC also ensure faster convergence compared to individual additive TMCMC and individual
multiplicative TMCMC.

\subsection{Additive-Multiplicative TMCMC}
\label{subsec:add_mult_tmcmc}
\ctn{Dutta13} described another TMCMC algorithm that uses the additive transformation
for some coordinates of $x$ and the multiplicative transformation for the remaining coordinates. 
\ctn{Dutta13} refer to this as additive-multiplicative TMCMC. Let the target density $\pi$ be
supported on $\mathbb R^d$. Then, if the additive transformation is used for the $i$-th coordinate, 
we update $x_i$
to $x_i+b_i\epsilon_1$, where $b_i$ is defined by (\ref{eq:b_add}), and $\epsilon\sim g^{(1)}$.
On the other hand, if for any coordinate $x_j$, the multiplicative transformation is used,
then we simulate $b_j$ following (\ref{eq:b_mult}), simulate $\epsilon_2\sim g^{(2)}$,
and update $x_j$ to either $x_j\epsilon_2$ or $x_j/\epsilon_2$ accordingly as $b_j=+1$ or $-1$.
If $b_j=0$, then we leave $x_j$ unchanged. The new proposal is accepted with probability 
having the same form as (\ref{eq:acc_multiplicative}). Note that unlike the cases of additive
TMCMC and multiplicative TMCMC, which use a single $\epsilon$ to update all the $d$ coordinates
of $x$, here we need two $\epsilon$'s: $\epsilon_1$ and $\epsilon_2$, to update the $d$ coordinates.

The proof of geometric ergodicity of additive-multiplicative TMCMC is almost the same as
that of multiplicative TMCMC, and hence we omit it from this paper.

\subsection{Geometric ergodicity}


Let $P$ be the transition kernel of a $\psi$-irreducible, aperiodic, positive Harris recurrent Markov 
chain with the stationary distribution $\pi$. 
Then the chain is geometrically ergodic if there exist a function $V \geq 1$ which is finite at least at one point, 
and constants $0<\rho<1$ and $M<\infty$ satisfying

\begin{equation}\label{eq:geo}
\|P^{n}(x,\cdot)- \pi(\cdot)\|_{TV} \leq MV(x)\rho^{n}  \hspace{0.5 cm} \forall n \geq 1,
\end{equation}
where $\|\mu\|_{TV}=\underset{g: |g| \leq V}{\sup} \mu(g)$ denotes the \emph{total variation norm}.
A standard way of checking geometric ergodicity is a result that involves small sets and the `geometric drift condition'. 
A set $E$ is called small if there exists $m>0$, $\delta>0$ and a probability measure $\nu$ such that
for $x\in E$,
\begin{equation}
P^m(x,\cdot)\geq\delta\nu(\cdot).
\end{equation}

$P$ is said to have geometric drift to a small set $E$ if there is a function $V \geq 1$, finite for at least one point,
and constants $\lambda <1 $ and $\zeta < \infty $ so that

\begin{equation}\label{eq:drift}
PV(x) \leq \lambda V(x) + \zeta I_{E}(x), 
\end{equation}
where $PV(x)= \int{V(y)P(x,y)dy}$ is the expectation of $V$ after one transition 
given that one starts at the point $x$, and $I_{E}(x)=1$ if $x\in E$ and $0$ otherwise, is the indicator function. 
Theorems 14.0.1 and 15.0.1 in \ctn{Meyn93} establish the fact that if $P$ has a geometric drift to a small set $E$, 
then under certain regularity conditions, $P$ is $\pi$-almost everywhere geometrically ergodic and the 
converse is also true.

We now provide necessary and sufficient conditions in favour of (\ref{eq:drift}); the result can
be thought of as an adaptation of Lemma 3.5 of \ctn{Jarner00}.

\begin{lemma}\label{Lemma 1}
Assume that the Markov transition kernel $P$ is associated with additive, multiplicative, or
additive-multiplicative TMCMC.
If there exists $V$ such that $V \geq 1$ and finite on bounded support, such that the following hold
\begin{equation}\label{eq:1con}
\underset{\|x\| \rightarrow \infty}{\lim\sup}~{\frac{PV(x)}{V(x)}}  <  1\quad\mbox{and} 
\end{equation}
\begin{equation}{\label{eq:infcon}}
{\frac{PV(x)}{V(x)}} <  \infty  \hspace{1 cm} \forall x,
\end{equation}
then $V$ satisfies the geometric drift condition (\ref{eq:drift}) and hence the chain must be geometrically
ergodic. Also, if for some $V$ finite, the geometric drift condition is satisfied, then the above 
conditions must also hold true.
\end{lemma}

\begin{proof}
Assume that for some $V$ finite and $V \geq 1$, the geometric drift condition (\ref{eq:drift}) is satisfied. 
Now, dividing both sides by $V(x)$, we get
$$ \frac{PV(x)}{V(x)} \leq \lambda + \zeta \frac{I_{E}(x)}{V(x)}. $$
Since $V$ is finite, then given that $V \geq 1$, we have
$$ \frac{PV(x)}{V(x)} \leq  \lambda + \zeta \hspace { 0.2 cm}  <  \infty. $$
Also if $\|x\| \rightarrow \infty $ then as $E$ is a bounded small set, $I_{E}(x) \rightarrow 0$, and hence 
$$ \underset{\|x\| \rightarrow \infty}{\lim\sup}~{\frac{PV(x)}{V(x)}} \leq \lambda < 1. $$

For the converse, let us fix a value $\gamma <1$. Let $R$ be particularly large so that if $\|x\| > R$, then 
$$ {\frac{PV(x)}{V(x)}} < \gamma\quad\mbox{if}\quad \|x\|> R  \hspace{0.2 cm}  
\implies PV(x)<  \gamma V(x)\quad\mbox{if}\quad \|x\| > R.  $$

Since
$$ PV(x) \leq \frac{PV(x)}{V(x)} V(x), $$
and since $\frac{PV(x)}{V(x)}$ is finite by hypothesis (\ref{eq:infcon}) and the function 
$V$ is also finite on any bounded set,
this implies that $PV(x)$ is finite on $E=\{x:\|x\| \leq R \}$, which is closed and bounded.

%

Take $\zeta$ to be the maximum value (which must be finite) that $PV(x)$ can attain on the set $E$. 
In the supplement of \ctn{Dutta13} it is shown that for additive TMCMC, sets of the form 
$E=\{x:\|x\| \leq R \}$ are small.
Defining
\begin{equation}
\mathcal V=\{(v_1,\ldots,v_d)\in\mathbb R^d: v_i=0~ \mbox{for at least one}~i\in\{1,\ldots,d\}\},
\label{eq:set_V}
\end{equation}
in the Appendix we will show that compact subsets of 
$\mathbb R^d\backslash\mathcal V$, which we denote by
$E^*$, are small 
for multiplicative TMCMC; the same result also holds for additive-multiplicative TMCMC. 
Hence, for all $x$, if $\mathbb E$ is either $E$ or $E^*$,
$$  PV(x) \leq \gamma V(x)+ \zeta I_{\mathbb E}(x). $$
This proves the lemma.
\end{proof}

So, in order to check geometric ergodicity, it is enough to prove (\ref{eq:1con}) and (\ref{eq:infcon}) 
for the given chain.

\newcommand{\A}{A^{I_{k}}}
\newcommand{\R}{ R^{I_{k}}}
\newcommand{\xf}{{x_{I_{k}}}(\epsilon)}
\newcommand*{\bigleft}{\mbox{\Huge $[$}}
\newcommand*{\bigright}{\mbox{\Huge $]$}}
\newcommand*{\bigcurlleft}{\mbox{\huge $\{$}}
\newcommand*{\bigcurlright}{\mbox{\huge $\}$}}

\section{Geometric ergodicity of additive TMCMC}
\label{sec:geo_additive}

We shall now provide necessary and sufficient conditions for geometric ergodicity for additive TMCMC for a broad class of distributions. 
This proof follows on the lines of
\ctn{Jarner00} and has been suitably modified for our additive TMCMC case. First, we define the notion of super-exponential densities.

A density $\pi$ is said to be super-exponential if it is positive with continuous first derivative and satisfies 

\begin{equation}
\underset{\|x\| \rightarrow \infty}{\lim} n(x)' \nabla \log \pi(x) = -\infty,
\end{equation}
where $n(x)$ denotes the unit vector $ \frac{x}{\|x\|} $. This would imply that for any $K > 0$, $\exists$ $R > 0$ such that

\begin{equation}\label{eq:exp}
\frac{\pi(x+cn(x))}{\pi(x)} \leq e^{-cK}; \hspace{1 cm} \|x\|\geq R, c\geq 0.
\end{equation}

In words, the above definition entails that $\pi$ is decaying at a rate faster than exponential along any direction. 
It is very easy to check that the Gaussian (univariate as well as multivariate for any variance covariance matrix) 
or the Gamma distributions (univariate or independent multivariate) indeed satisfy these conditions. 

Let the acceptance region and the (potential) rejection region corresponding to the move-type $b$ be defined by
$A^{(1)}(b,x)=\{\epsilon:\pi(T^{(1)}_b(x,\epsilon))\geq \pi(x)\}$ and 
$R^{(1)}(b,x)=\{\epsilon:\pi(T^{(1)}_b(x,\epsilon))<\pi(x)\}$, respectively.
Also, let $A^{(1)}(x)=\cup_{b_1,\ldots,b_d}A^{(1)}(b,x)$ and $R^{(1)}(x)=\cap_{b_1,\ldots,b_d}R^{(1)}(b,x)$ denote the overall acceptance
region and the overall potential rejection region, respectively.

Let $Q^{(1)}(x,B)$ denote the probability corresponding to the additive TMCMC proposal 
of reaching the Borel set $B$ from $x$ in one step. Let $P^{(1)}$ denote the Markov transition kernel associated with
additive TMCMC.
Then the following theorem establishes geometric ergodicity of additive TMCMC in the super-exponential set-up.

\begin{theorem}\label{theorem:geo_additive}

If the target density $\pi$ is super-exponential and has contours that are nowhere piecewise parallel to 
$\{x:|x_1|=|x_2|=\cdots = |x_d|\}$, 
then the additive TMCMC chain satisfies geometric drift if and only if

\begin{equation}
\underset{\|x\| \rightarrow \infty}{\lim\inf}~ Q^{(1)}(x, A^{(1)}(x)) > 0.
\label{eq:liminf_Q_additive}
\end{equation}


\end{theorem}

\begin{proof}

Following the notation of \ctn{Jarner00}, let $C_{\pi(x)}$ be the contour of the density $\pi$ 
corresponding to the value $\pi(x)$. 
We define the radial cone 
$C_{\pi(x)}(\delta)$ around $ C_{\pi(x)}$ to be

\begin{equation}
C_{\pi(x)} (\delta) = \left \{ y + sn(y) : y \in C_{\pi(x)}, -\delta < s < \delta\right \}.
\label{eq:radial_cone}
\end{equation}
See Figure 1 of \ctn{Jarner00} for visualizing these regions in two dimensions. 

By (\ref{eq:liminf_Q_additive}) there exists a $\eta > 0$ such that 

\begin{equation}
\underset{\|x\| \rightarrow \infty} {\lim\sup} \hspace{0.2 cm} Q^{(1)} (x, R^{(1)}(x)) \leq 1-2\eta^{\frac{1}{2}}.
\label{eq:limsup_Q_additive}
\end{equation}

Take the belt length $\delta$ such that the probability that a move from $x$, the starting point, falls within this $\delta$ belt is 
less than $\eta$. 
That it is possible can be seen as follows. Note that there exists a compact set $E$ such that 
\begin{equation}
Q^{(1)}(x,E^{c}) <\frac{\eta}{2}. 
\label{eq:Q1}
\end{equation}
So, for given $\delta$, if we can ensure that our proposal distribution satisfies 
\begin{equation}
Q^{(1)}(x, C_{\pi(x)}(\delta) \cap E) < \frac{\eta}{2},
\label{eq:Q2}
\end{equation}
then we are done. 
Note that for any point on the contour, the probability that the additive TMCMC moves result 
in a value within $C_{\pi(x)}(\delta)$ is bounded above by $2c\delta$, for some finite $c$ (since this probability
is $2\int_0^{\delta}g(\epsilon)d\epsilon\leq 2c\delta$, as  $g(\epsilon)\leq c$ on $(0,\delta)$, for $0<c<\infty$) 
and thus can be made as small as desired by choosing $\delta$ sufficiently small.
The above argument is easy to visualize in two dimensions as depicted in  
Figure \ref{fig:addTMCMC} and Figure 1 of \ctn{Jarner00} -- for any point
in the first quadrant part of the contour, the probability that the outer and inner TMCMC moves given, respectively, 
by $(+\epsilon, +\epsilon)$ and $(-\epsilon,-\epsilon)$, land within $C_{\pi(x)}(\delta)$, is bounded above by $2c\delta$.
The same argument applies to the other three quadrants.
For the other moves, note that since the contours (intersected with $E$) are nowhere piecewise parallel to 
$\{x:|x_1|=\cdots = |x_d|\}$, 
the moves can fall in only finite number of regions of $C_{\pi(x)}(\delta)\cap E$.
Infinitely many regions can be ruled out because of the intersection with $E$, which is compact. 
If that was the case, then this infinite collection of interesting points would 
have a limit point in $E$, which is not possible as the points are isolated. 

Now, there exists $R_{\eta}$ so that for any point $y$ outside the $\delta$ bound around $x$ and in the rejection region, 
it holds that
\begin{equation}
\frac{\pi(y)}{\pi(x)} < \eta; \hspace {0.5 cm} \|x\| > R_{\eta}.
\label{eq:super_exp}
\end{equation}
This can be seen by taking the shortest line from $y$ to the origin; suppose it intersects 
(after extending if needed) the contour $C_{\pi(x)}$ at $z$. There will be two such values of $z$, 
and we choose the one that is nearest to $x$. 
Then, by (\ref{eq:exp}) and the fact that $ \pi(x)$ is the same as $\pi(z)$ (since $x$ and $z$ are on the same contour), 
we obtain (\ref{eq:super_exp}). 
To ensure that this $z$ indeed satisfies $ \|z\| > R_{\eta}$, consider the set $E$, which is 
the set where effectively all the moves fall. 
Join each point in $E$ to the origin by a straight line and extend it if needed to intersect the contour;
consider those points of intersections which are closest to $x$. 
The points of intersections yield a segment, $D(x)$, of the contour, which contains $x$ and is bounded and closed. 
Now since this set is bounded, 
we can always choose $x$ with large enough norm so that all the points in $E$ associated with $D(x)$ 
have norms greater than $R_{\eta}$. Since $z$ is one of such points, we are done.

On the other hand, if $y$ is outside the $\delta$ bound around $x$ but falls in the acceptance region,
then by the same arguments it holds that
\begin{equation}
\frac{\pi(x)}{\pi(y)} < \eta. 
\label{eq:super_exp2}
\end{equation}

Now, with $ V(x) = \frac{c}{\sqrt{\pi(x)}}$ for some $c > 0$ chosen appropriately, 
it holds that

\begin{eqnarray}
&&\frac{P^{(1)}V(x)}{V(x)}\nonumber\\
&= & \frac{1}{2^{d}} \sum_{b_{1},\cdots, b_{d}} 
\int _{A^{(1)}(x)} {\left [ \frac{\pi(x_{1},\ldots,x_{d})}{\pi(x_{1}+b_{1}\epsilon,\ldots, x_{d}+b_{d}\epsilon)} \right ]^{\frac{1}{2}} 
g^{(1)}(\epsilon)d\epsilon} \nonumber\\
&& + \frac{1}{2^{d}} \sum_{b_{1},\cdots, b_{d}} \int _{R^{(1)}(x)}{ 
\left [ 1- \frac{\pi(x_{1}+b_{1}\epsilon,\ldots,x_{d}+b_{d}\epsilon)}{\pi(x_{1},\ldots,x_{d})} 
+ \left \{ \frac{\pi(x_{1}+b_{1}\epsilon,\ldots, x_{d}+b_{d}\epsilon)}{\pi(x_{1},\ldots,x_{d})} \right \}^{\frac{1}{2}} \right ] 
g^{(1)}(\epsilon)d\epsilon}\nonumber\\
\label{eq:sum_ratios_additive}
\end{eqnarray}

We split the integrals over $A^{(1)}(x)$ and that over $R^{(1)}(x)$ into two parts -- 
within $C_{\pi(x)}(\delta)$ and outside $C_{\pi(x)}(\delta)$. 
%
Since $\frac{\pi(x_{1},\ldots,x_{d})}{\pi(x_{1}+b_{1}\epsilon,\ldots, x_{d}+b_{d}\epsilon)}<1$ on $A^{(1)}(x)$, it follows
from (\ref{eq:Q1}) and (\ref{eq:Q2}) that
\begin{align}
\int _{A^{(1)}(x)\cap C_{\pi(x)}(\delta)} 
{\left [ \frac{\pi(x_{1},\ldots,x_{d})}{\pi(x_{1}+b_{1}\epsilon,\ldots, x_{d}+b_{d}\epsilon)} \right ]^{\frac{1}{2}}
g^{(1)}(\epsilon)d\epsilon}
&<\frac{\eta}{2}.
\label{eq:part1}
\end{align}

Now note that 
\begin{align}
\left\vert\frac{\|x_n+b\epsilon\|^2}{\|x_n\|^2}-1\right\vert 
&\leq\frac{2\epsilon\left|b'x_n\right|}{\|x_n\|^2}+\frac{d\epsilon^2}{\|x_n\|^2}\notag\\
&\leq \frac{2\epsilon}{\|x_n\|}+\frac{d\epsilon^2}{\|x_n\|^2}.\label{eq:as_conv}\\
& (\mbox{since}~ \left|b'x_n\right|\leq\sqrt{\|b\|^2\|x_n\|^2}=\|x_n\|).\notag
\end{align}

Let $\mathcal N_{\epsilon}$ denote a null set associated with the probability distribution of $\epsilon$.
Then for all $\omega\in\mathcal N^c_{\epsilon}$ such that $\epsilon(\omega)\in E$, for any compact set
$E$, (\ref{eq:as_conv}) goes to zero. That is, for $\omega\in\mathcal N^c_{\epsilon}\cap\epsilon^{-1}(E)$,
$\frac{\|x_n+b\epsilon\|}{\|x_n\|}\rightarrow 1$.
Thus, for $n>N_0(\eta_2(\omega))$ for some $N_0(\eta_2(\omega))$ depending upon $\eta_2(\omega)$ such that 
$\|x_{n}\|>\frac{R_{\eta}}{1-\eta_2(\omega)}>R_{\eta}$,
since $1+\eta_2(\omega)>\frac{\|x_{n}+b\epsilon\|}{\|x_{n}\|}>1-\eta_2(\omega)$ for 
$\omega\in\mathcal N^c_{\epsilon}\cap\epsilon^{-1}(E)$, we have 
\begin{equation}
\|x_{n}+b\epsilon\|>(1-\eta_2(\omega))\|x_{n}\|>R_{\eta}.
\label{eq:p0}
\end{equation}
Note that for any given $\zeta>0$, we can choose $E_{\zeta}$ such that $Q^{(1)}(x,E^c_{\zeta,x,b})<\zeta$, 
for any $x$ and $b$, where $E^c_{\zeta,x,b}=\{x+b\epsilon:\epsilon\in E^c_{\zeta}\}$. 
Thus, we can choose $\zeta>0$ such that
\begin{equation}
Q^{(1)}(x,A^{(1)}(x)\cap C^c_{\pi(x)}(\delta)\cap E^c_{\zeta,x,b})<\frac{\eta^{\frac{1}{2}}}{2}Q^{(1)}(x,A^{(1)}(x)).
\label{eq:p1}
\end{equation}

Now, 
it follows from (\ref{eq:p0}) and (\ref{eq:super_exp2}) that 
for given $\eta>0$, we can choose $R_{\eta}$ such that for $\|x\|>R_{\eta}$, 
$\frac{\pi(x_{1},\ldots,x_{d})}{\pi(x_{1}+b_{1}\epsilon,\ldots, x_{d}+b_{d}\epsilon)}<\frac{\eta}{4}$.
Hence,
\begin{align}
\int _{A^{(1)}(x)\cap C^c_{\pi(x)}(\delta)\cap E_{\zeta,x,b}} 
{\left [ \frac{\pi(x_{1},\ldots,x_{d})}{\pi(x_{1}+b_{1}\epsilon,\ldots, x_{d}+b_{d}\epsilon)} \right ]^{\frac{1}{2}}
g^{(1)}(\epsilon)d\epsilon}
&<\frac{\eta^{\frac{1}{2}}}{2}Q^{(1)}(x,A^{(1)}(x)).
\label{eq:part2_1}
\end{align}
Also, since $\frac{\pi(x_{1},\ldots,x_{d})}{\pi(x_{1}+b_{1}\epsilon,\ldots, x_{d}+b_{d}\epsilon)}<1$ on
$A^{(1)}(x)$, it follows from (\ref{eq:p1}) that
\begin{align}
\int _{A^{(1)}(x)\cap C^c_{\pi(x)}(\delta)\cap E^c_{\zeta,x,b}} 
{\left [ \frac{\pi(x_{1},\ldots,x_{d})}{\pi(x_{1}+b_{1}\epsilon,\ldots, x_{d}+b_{d}\epsilon)} \right ]^{\frac{1}{2}}
g^{(1)}(\epsilon)d\epsilon}
&<\frac{\eta^{\frac{1}{2}}}{2}Q^{(1)}(x,A^{(1)}(x)).
\label{eq:part2_2}
\end{align}
Thus, for $\|x\|>R_{\eta}$, it follows from (\ref{eq:part2_1}) and (\ref{eq:part2_2}) that  
\begin{align}
\int _{A^{(1)}(x)\cap C^c_{\pi(x)}(\delta)} 
{\left [ \frac{\pi(x_{1},\ldots,x_{d})}{\pi(x_{1}+b_{1}\epsilon,\ldots, x_{d}+b_{d}\epsilon)} \right ]^{\frac{1}{2}}
g^{(1)}(\epsilon)d\epsilon}
&<\eta^{\frac{1}{2}}Q^{(1)}(x,A^{(1)}(x)).
\label{eq:part2}
\end{align}

Now note that on $R^{(1)}(x)$, $1- \frac{\pi(x_{1}+b_{1}\epsilon,\ldots,x_{d}+b_{d}\epsilon)}{\pi(x_{1},\ldots,x_{d})}<1$, so that  
\begin{equation}
\int_{R^{(1)}(x)}\left[1- \frac{\pi(x_{1}+b_{1}\epsilon,\ldots,x_{d}+b_{d}\epsilon)}{\pi(x_{1},\ldots,x_{d})}\right]
g^{(1)}(\epsilon)d\epsilon<Q^{(1)}(x,R^{(1)}(x)).
\label{eq:part3}
\end{equation}
For the integral $\int_{R^{(1)}(x)}\frac{\pi(x_{1}+b_{1}\epsilon,\ldots, x_{d}+b_{d}\epsilon)}{\pi(x_{1},\ldots,x_{d})}
g^{(1)}(\epsilon)d\epsilon$,
breaking up $R^{(1)}(x)$ into $R^{(1)}(x)\cap C_{\pi(x)}(\delta)$ and $R^{(1)}(x)\cap C^c_{\pi(x)}(\delta)$
we obtain, in exactly the same way as (\ref{eq:part1}) and (\ref{eq:part2}), the following:
\begin{equation}
\int_{R^{(1)}(x)}\left[\frac{\pi(x_{1}+b_{1}\epsilon,\ldots, x_{d}+b_{d}\epsilon)}{\pi(x_{1},\ldots,x_{d})}\right]^{\frac{1}{2}}
g^{(1)}(\epsilon)d\epsilon
<\frac{\eta}{2}+\eta^{\frac{1}{2}}Q^{(1)}(x,R^{(1)}(x)),
\label{eq:part4}
\end{equation}
for $ \|x\| > R_{\eta}$. Combining (\ref{eq:part1}), (\ref{eq:part2}), (\ref{eq:part3}) and (\ref{eq:part4}) we obtain

\begin{eqnarray}
\frac{P^{(1)}V(x)}{V(x)} &< & \eta+\eta^{\frac{1}{2}} Q^{(1)}(x, A^{(1)}(x)) + \left (1+\eta^{\frac{1}{2}} \right ) Q^{(1)} (x, R^{(1)}(x)) \nonumber \\
& =& \eta +\eta^{\frac{1}{2}} + Q^{(1)} (x, R^{(1)}(x)). \nonumber \\
\end{eqnarray}
Using (\ref{eq:limsup_Q_additive}), we obtain
\begin{eqnarray}
\underset{\|x\| \rightarrow \infty} {\lim\sup}~\frac{PV(x)}{V(x)}&< & \eta +\eta^{\frac{1}{2}} 
+ \underset{\|x\| \rightarrow \infty} {\lim\sup}~Q^{(1)} (x, R^{(1)}(x))\nonumber\\
& =& 1- \eta^{\frac{1}{2}} +\eta\nonumber\\
&<& 1.\nonumber
\end{eqnarray}
Thus, (\ref{eq:1con}) is satisfied. Since all the ratios in the integrals of (\ref{eq:sum_ratios_additive}) are less than 1,
it is clear that $P^{(1)}V(x)/V(x)<\infty$ for all $x$, satisfying (\ref{eq:infcon}). This proves geometric
ergodicity of additive TMCMC.

Now we prove that if additive TMCMC is geometrically ergodic, then (\ref{eq:liminf_Q_additive}) is satisfied. 
In fact, we prove that if (\ref{eq:liminf_Q_additive}) is not satisfied, that is, if
$\underset{\|x\| \rightarrow \infty} {\lim\sup}~Q^{(1)} (x, R^{(1)}(x))=1$, then 
$\underset{\|x\| \rightarrow \infty} {\lim\sup}~P^{(1)} (x, \{x\})=1$. Indeed,
it follows from Theorem 5.1 of \ctn{Roberts96} that the latter condition 
implies that $P^{(1)}$ is not geometrically ergodic.

We can choose a compact set $E$ such that $Q^{(1)}(x,E^c)<\eta$ and can choose $\delta$ small
enough such that $\underset{\|x\| \rightarrow \infty} {\lim\sup}~Q^{(1)} (x,C_{\pi(x)}(\delta)\cap E)\leq\eta$. 
This and the fact (\ref{eq:super_exp}) 
imply that 
\begin{eqnarray}
\underset{\|x\| \rightarrow \infty} {\lim\sup}~P^{(1)}(x,\{x\}) 
&\geq & 
\underset{\|x\| \rightarrow \infty} {\lim\sup}~
\frac{1}{2^{d}} \sum_{b_{1},\cdots, b_{d}}\int _{R^{(1)}(x)}{ 
\left [ 1- \frac{\pi(x_{1}+b_{1}\epsilon,\ldots,x_{d}+b_{d}\epsilon)}{\pi(x_{1},\ldots,x_{d})} 
\right ] 
g^{(1)}(\epsilon)d\epsilon}\nonumber\\
& \geq & 
\underset{\|x\| \rightarrow \infty} {\lim\sup}~
\frac{1}{2^{d}} \sum_{b_{1},\cdots, b_{d}} \int _{R^{(1)}(x)\cap E\cap [C_{\pi(x)}(\delta)]^c}{ 
\left [ 1- \frac{\pi(x_{1}+b_{1}\epsilon,\ldots,x_{d}+b_{d}\epsilon)}{\pi(x_{1},\ldots,x_{d})} 
\right ] 
g^{(1)}(\epsilon)d\epsilon}\nonumber\\
&\geq & (1-\eta)
\underset{\|x\| \rightarrow \infty} {\lim\sup}~
Q(x,R^{(1)}(x)\cap E\cap [C_{\pi(x)}(\delta)]^c)\nonumber\\
&\geq & (1-\eta)(1-2\eta).\notag
\end{eqnarray}
Since $\eta>0$ is arbitrary, the proof is complete.
\end{proof}

\begin{figure}
\centering
\includegraphics[width=12cm,height=7cm]{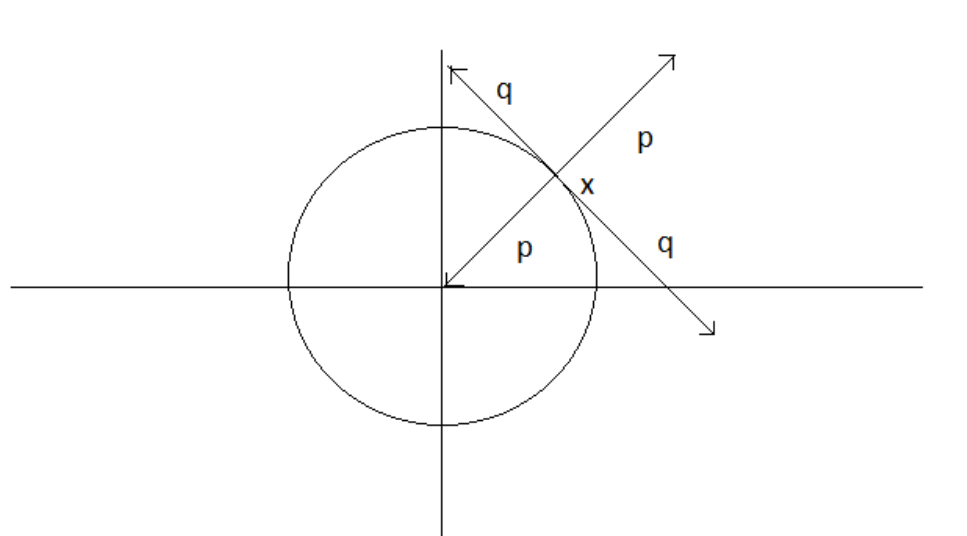}
\caption{A contour of a spherically symmetric distribution. Here $x$ is the current state lying
on the contour (first quadrant), and the four directions that can be taken by the next move of additive TMCMC, are
displayed. Here $p=q=1/2$ are the move-type probabilities.}
\label{fig:norm}
\end{figure}
Note that for spherically symmetric super-exponential distributions (for example standard Gaussian), 
the conditions of Theorem \ref{theorem:geo_additive} naturally hold. For instance, the fact that no part of the contour is 
parallel to $\{x:|x_1|=|x_2|=\cdots = |x_d|\}$ is quite obvious. 
To check that $\underset{\|x\| \rightarrow \infty}{\lim\inf}~Q^{(1)}(x, A^{(1)}(x)) > 0$, 
first perceive (see Figure \ref{fig:norm}) that at any point in the first quadrant, the inward 
direction stays in the acceptance region if the magnitude of the inward direction does not
exceed the diameter of the contour containing $x$. However, the inward direction  
can land in the rejection region on the other side of the contour if the magnitude of the inward
direction exceeds the diameter of the contour $C_{\pi(x)}$. Since $\|x\|$ is the radius of
$C_{\pi(x)}$, in order to ensure that the inward move falls in $A^{(1)}(x)$ with high probability 
when $\|x\|$ is large, 
we must choose the proposal density $g^{(1)}(\epsilon)$ in such a way that too large step sizes compared
to $\|x\|$, have small probabilities. 
Thus, for our purpose, first let $M_{\eta}$ be such that 
$ \int_{0}^{M_{\eta}}{g^{(1)}(\epsilon)d\epsilon} > 1-\eta$. 
Now choose $x$ such that $\|x\| > 3 M_{\eta}$ (radius of $C_{\pi(x)}$ is greater than $3 M_{\eta})$. 
Then $ Q^{(1)}(x,A^{(1)}(x)) > \frac{1-\eta}{4} > 0$. Now consider any sequence ${x_
{n}}$ with $ \|x_{n}\| \rightarrow \infty$, where $x_{n}$ has norm greater than $ 3 M_{\eta}$
for all but finite $n$. 
Then along this sequence, 
the limit of $ Q^{(1)}(x, A^{(1)}(x))$ is greater than $ \frac{1-\eta}{4}$. 
Thus $\underset{\|x\| \rightarrow \infty}{\lim\inf}~Q^{(1)}(x, A^{(1)}(x)) > 0$ condition is satisfied.

Note that the constraint that no part of the contour can be piecewise parallel to $\{x:|x_1|=\cdots = |x_d|\}$ does not really 
cause too much of a problem because the only common distribution that satisfies this property is the Laplace 
distribution and it is not super-exponential.

\section{Geometric ergodicity of multiplicative TMCMC}
\label{sec:geo_multiplicative}

In the one-dimensional case, geometric ergodicity of multiplicative TMCMC has been established by \ctn{Dutta12},
assuming that the target density is regularly varying in an appropriate sense.
Here we extend the result to arbitrary dimensions, of course without the aid of the 
regularly varying assumption, since such an assumption is not well-defined in high dimensions.
Note however, that since vectors $v\in\mathcal V$ (where $\mathcal V$ is defined in (\ref{eq:set_V})), 
can not belong to small sets associated with multiplicative TMCMC, to prove
geometric ergodicity we also need to show that $\underset{\|x-v\|\rightarrow 0}{\lim\sup}~P^{(2)}V(x)/V(x)<1$
for all $v\in\mathcal V$. This seems to be too demanding a requirement.
In the one-dimensional case, $0$ is the only point which can not belong to small sets, and the proof
of geometric ergodicity in this case requires showing $\underset{|x|\rightarrow 0}{\lim\sup}~P^{(2)}V(x)/V(x)<1$.
This has been established by \ctn{Dutta12}, however, the technique of his proof could not assist us in our
complicated, high-dimensional case.

If one has the liberty to assume, in the high-dimensional case, that 
there is an arbitrarily small, compact neighborhood $\mathbb N_0$,
of $\bzero=(0,0,\ldots,0)'$, which has zero probability under the target density $\pi$,
then the proof of 
$\underset{\|x-v\|\rightarrow 0}{\lim\sup} P^{(2)}V(x)/V(x)<1$ for all $v\in\mathcal V$
is not required.
Although for practical purposes this is not a very stringent assumption, from the theoretical standpoint this
is somewhat disconcerting. In the next subsections we introduce two different kinds of 
geometric ergodic mixtures of additive and multiplicative TMCMC kernels that do not require
the undesirable assumption $\pi(\mathbb N_0)=0$. The first mixture we introduce is essentially multiplicative
TMCMC in a sense to be made precise subsequently, whereas the second mixture is a straightforward
convex combination of additive and multiplicative TMCMC kernels.

\subsection{A new mixture-based Markov transition kernel with ``essentially full" weight
on multiplicative TMCMC}
\label{subsec:new_kernel}

We break up $\pi$
into a mixture of 
two densities: $\pi_1$, supported on $\mathbb N_0$, 
and $\pi_2$, supported on $\mathbb N^c_0$.
That is, we write
\begin{align}
\pi(x)&=\pi(\mathbb N_0)\frac{\pi(x)}{\pi(\mathbb N_0)}I\{x\in\mathbb N_0\}
+\pi(\mathbb N^c_0)\frac{\pi(x)}{\pi(\mathbb N^c_0)}I\{x\in\mathbb N^c_0\}\notag\\
&=\pi(\mathbb N_0)\pi_1(x)+\pi(\mathbb N^c_0)\pi_2(x),\label{eq:mixture}
\end{align}
where 
\begin{align}
\pi_1(x)&=\frac{\pi(x)}{\pi(\mathbb N_0)}I\{x\in\mathbb N_0\}\quad\mbox{and}\quad
\pi_2(x)=\frac{\pi(x)}{\pi(\mathbb N^c_0)}I\{x\in\mathbb N^c_0\}.\label{eq:pi_split}
\end{align}
Clearly, $\pi_2(\mathbb N_0)=0$. In fact, as we elaborate below, the above mixture representation transfers the 
requirement $\pi(\mathbb N_0)=0$
to $\pi_2(\mathbb N_0)=0$.

Now consider the following Markov chain: for any $x\in\mathbb R^d$ and $A\in\mathcal B(\mathbb R^d)$,
with $\mathcal B(\mathbb R^d)$ being the Borel $\sigma$-field of $\mathbb R^d$,
\begin{equation}
P(x,A)=\pi(\mathbb N_0)P^{(1)}(x,A)+\pi(\mathbb N^c_0)P^{(2)}(x,A),
\label{eq:mc1}
\end{equation}
where $P^{(1)}(x,\cdot)$ and $P^{(2)}(x,\cdot)$ are Markov transition kernels corresponding
to additive TMCMC converging to $\pi_1$ and multiplicative TMCMC converging to $\pi_2$, respectively. 
We choose the proposal density $g^{(2)}$ 
for multiplicative TMCMC such that there is a one-dimensional, arbitrarily small neighborhood of $0$ which receives
zero probability under $g^{(2)}$. We denote the one-dimensional neighborhood of $0$ by $\mathcal N_0$.
We also assume that there exist 
arbitrarily small neighborhoods $\mathcal N_{+1}$ and $\mathcal N_{-1}$
of $+1$ and $-1$ respectively, which receive zero probability under $g^{(2)}$. 

In order to implement the mixture kernel $P$, we can separately run two 
chains -- one is additive TMCMC converging to $\pi_1$ and another is
multiplicative TMCMC converging to $\pi_2$, both chains starting at the
same initial value $x_0$. Since both the chains are positive Harris recurrent on 
$\mathbb R^d\backslash\mathcal V$ 
($P^{(1)}$ is positive Harris recurrent on $\mathbb R^d$ and $P^{(2)}$ is positive Harris recurrent on 
$\mathbb R^d\backslash\mathcal V$)
convergence to both
$\pi_1$ and $\pi_2$ occurs for the initial value $x_0~(\neq\bzero)$, even though the supports of
$\pi_1$ and $\pi_2$ are disjoint. In practice, it will be convenient to choose $x_0$
from the boundary between $\mathbb N_0$ and $\mathbb N^c_0$. Thus, for any initial value
$x_0$, we will have an additive TMCMC chain $\{x^{(k)}_1;k=0,1,2,\ldots\}$ converging to $\pi_1$
and another multiplicative TMCMC chain $\{x^{(k)}_2;k=0,1,2,\ldots\}$ converging to $\pi_2$,
with $x^{(0)}_1=x^{(0)}_2=x_0$. 

Finally, for each $k=1,2,\ldots$, we select and store
$x^{(k)}_1$ with probability $\pi(\mathbb N_0)$ and $x^{(k)}_2$ with probability $[1-\pi(\mathbb N_0)]$. 
Thus, for $k>1$, the chain $\left\{P^{(1)}\right\}^k$ depends only on $x^{(k-1)}_1$, and not
on $x^{(k-1)}_2$. Similarly, $\left\{P^{(2)}\right\}^k$ depends only on $x^{(k-1)}_2$
and not on $x^{(k-1)}_1$. 

Thus, the mixture $P$ uses additive TMCMC to simulate only from $\pi_1$, and uses
multiplicative TMCMC to simulate only from $\pi_2$. Since $\pi(\mathbb N_0)$ is negligibly small,
the mixture $P$ gives ``essentially full" weight to multiplicative TMCMC.

If we can prove that $P^{(2)}$ is geometrically ergodic for $\pi_2$, then because $P^{(1)}$
is geometrically ergodic for $\pi_1$ (in fact, uniformly ergodic for $\pi_1$ 
since the support $\mathbb N_0$ of $\pi_1$ is compact), it will follow that $P$ 
itself is geometrically ergodic. See Appendix \ref{sec:P_geo} for a proof
of this statement.

Note that $\pi(\mathbb N_0)$ is unknown and needs to be estimated for implementing $P$.
In Appendix \ref{sec:implement_P} we present an importance sampling based idea regarding this,
also demonstrating why the estimated probability is expected to yield the same TMCMC samples as
the exact value of $\pi(\mathbb N_0)$.

The mixture kernel $P$ given by (\ref{eq:mc1}) is designed to give almost full weight to
multiplicative TMCMC. It is also possible to consider a more conventional mixture of 
additive and multiplicative TMCMC, which is also geometrically ergodic but combines the 
good features of both the algorithms
to yield a more efficient TMCMC sampler, and does not require estimation of $\pi(\mathbb N_0)$.
In the next subsection 
we discuss this in detail, also elucidating
how (\ref{eq:mc1}) differs from the combination of
additive and multiplicative TMCMC in a traditional mixture set-up.

\subsection{Combination of additive and multiplicative TMCMC in a traditional mixture set-up}
\label{subsec:usual_mixture_kernel}

Instead of (\ref{eq:mc1}), we could define a mixture of the form
\begin{equation}
P^*(x,A)=p~P^{(1)}(x,A)+(1-p)~P^{(2)}(x,A),
\label{eq:mc2}
\end{equation}
where $0<p<1$ is any choice of mixing probability. The transition kernels $P^{(1)}(x,\cdot)$
and $P^{(2)}(x,\cdot)$, as before, are additive and multiplicative TMCMC, respectively,
but here each of them converges to the target density $\pi$, unlike the case of (\ref{eq:mc1}) where
$P^{(1)}(x,\cdot)$ converged to $\pi_1$ and $P^{(2)}(x,\cdot)$ converged to $\pi_2$.

For the implementation of $P^*$, one can first simulate $u\sim U(0,1)$; 
if $u<p$, then additive TMCMC will be implemented, else multiplicative TMCMC should be used. 
Thus, unlike the case of (\ref{eq:mc1}), we have a single chain $\{x^{(k)};k=0,1,2,\ldots\}$
converging to $\pi$. Note also, that $P^*$ implements 
both additive and multiplicative TMCMC on the entire support of $\pi$. In contrast, 
$P$, given by (\ref{eq:mc1}), implements additive TMCMC only for $\pi_1$, which is supported on
$\mathbb N_0$ and implements multiplicative TMCMC only for $\pi_2$, which is supported on $\mathbb N^c_0$.

In Section \ref{sec:geo_additive} we have already shown, for $V(x)=c/\sqrt{\pi(x)}$, that 
$\underset{\|x\|\rightarrow\infty}{\lim\sup}~\frac{P^{(1)}V(x)}{V(x)}<1$ and 
that the ratio $\frac{P^{(1)}V(x)}{V(x)}$ is finite for all $x$.
For the same function $V$ if we can also prove that 
$\underset{\|x\|\rightarrow\infty}{\lim\sup}~\frac{P^{(2)}V(x)}{V(x)}<1$, and that the ratio
$\frac{P^{(2)}V(x)}{V(x)}$ is finite for all $x$,
then it follows that the mixture $P^*$ is also geometrically ergodic. Indeed, 
\begin{align}
\underset{\|x\|\rightarrow\infty}{\lim\sup}~\frac{P^*V(x)}{V(x)} &\leq 
p~\underset{\|x\|\rightarrow\infty}{\lim\sup}~\frac{P^{(1)}V(x)}{V(x)}
+(1-p)~\underset{\|x\|\rightarrow\infty}{\lim\sup}~\frac{P^{(2)}V(x)}{V(x)}\notag\\
&<p+(1-p)=1,\notag
\end{align}
and $v\in\mathcal V$ can be a limit point of small sets corresponding to $P^*$, since
$P^*(x,A)\geq pP^{(1)}(x,A)$ for all $x$ and all $A\in\mathcal B (\mathbb R^d)$, and all compact sets of 
$\mathbb R^d$ are small sets of $P^{(1)}$.

\subsection{Distinctions between the roles of $P^{(2)}$ in $P$ and $P^*$}
\label{subsec:distinctions}

\subsubsection{Geometric ergodicity of $P$ requires geometric ergodicity of $P^{(2)}$}
\label{subsubsec:geo_P}
The proof of geometric ergodicity of the essentially fully multiplicative mixture $P$ will follow
if we can show that $P^{(2)}$ is geometrically ergodic for $\pi_2$, where $\pi_2(\mathbb N_0)=0$ by construction.
Theorem \ref{theorem:geo_multiplicative} provides necessary and sufficient conditions for
geometric ergodicity of $P^{(2)}$ under the super-exponential set-up, assuming $\pi_2(\mathbb N_0)=0$,
and that the proposal density $g^{(2)}$ gives zero probability to arbitrarily small compact neighborhoods of
$0$, $-1$ and $+1$. 
Since it is always possible to construct a proposal density $g^{(2)}$ with the requisite properties,
the strategy of forming the mixture $P$ is not restrictive, given the super-exponential set-up.

\subsubsection{Geometric ergodicity of $P^*$ does not require geometric ergodicity of $P^{(2)}$ or the restriction 
$\pi(\mathbb N_0)=0$}
\label{subsubsec:geo_P_star}
Geometric ergodicity of the traditional mixture $P^*$, on the other hand, follows only if 
$\underset{\|x\|\rightarrow\infty}{\lim\sup}~\frac{P^{(2)}V(x)}{V(x)}<1$ 
and $\frac{P^{(2)}V(x)}{V(x)}$ is finite for all $x$; it does not require $\pi(\mathbb N_0)=0$, $\pi$
being the invariant target distribution for both $P^{(1)}$ and $P^{(2)}$ of $P^*$. That the former two
conditions hold under the aforementioned assumptions, without the restriction $\pi(\mathbb N_0)=0$, can be easily seen
in the proof of Theorem \ref{theorem:geo_multiplicative}. 
The implication is that, super-exponentiality of $\pi$ and the aforementioned properties
of the proposal distribution $g^{(2)}$ guarantee geometric ergodicity of $P^*$, 
even though $P^{(2)}$ does not individually converge to its invariant distribution $\pi$ 
because the proposal density $g^{(2)}$ assigns 
zero probability to some arbitrarily small compact neighborhood of $\bzero$ (in the first place, the fact that
$P^*$ converges to $\pi$ is clear because of irreducibility, aperiodicity and positive Harris recurrence of $P^*$). 
As before, assumption of such proposal density $g^{(2)}$ is not restrictive,
and hence the strategy of forming the mixture $P^*$ is not restrictive either, given the super-exponential set-up.

In the next section we introduce our theorem characterizing geometric ergodicity of $P^{(2)}$.

\subsection{Geometric ergodicity of $P^{(2)}$ for $\pi_2$}
\label{subsec:geo_multiplicative}

For multiplicative TMCMC, for a given move-type $b$, we define the acceptance region and the 
potential rejection region by
$A^{(2)}(b,x)=\{\epsilon:\frac{\pi(T^{(2)}_b(\epsilon))}{\pi(x)}|J(b,\epsilon)|\geq 1\}$ and 
$R^{(2)}(b,x)=\{\epsilon:\frac{\pi(T^{(2)}_b(\epsilon))}{\pi(x)}|J(b,\epsilon)|< 1\}$, respectively.
The overall acceptance region and the overall potential rejection region are
$A^{(2)}(x)=\cup_{b_1,\ldots,b_d}A^{(2)}(b,x)$ and $R^{(2)}(x)=\cap_{b_1,\ldots,b_d}R^{(2)}(b,x)$, respectively.
We also define
$A^*(b,x)=\{\epsilon:\frac{\pi(T^{(2)}_b(\epsilon))}{\pi(x)}\geq 1\}$ and 
$R^*(b,x)=\{\epsilon:\frac{\pi(T^{(2)}_b(\epsilon))}{\pi(x)}< 1\}$, respectively.

Let $Q^{(2)}(x,B)$ denote the probability corresponding to the multiplicative TMCMC proposal 
of reaching the Borel set $B$ from $x$ in one step.

Then the following theorem characterizes geometric ergodicity of multiplicative TMCMC under 
the super-exponential set-up. For our convenience, we slightly abuse notation by referring to $\pi_2$
as $\pi$.

\begin{theorem}\label{theorem:geo_multiplicative}

Suppose that $\pi$, the target density, is super-exponential and has contours that are nowhere piecewise parallel 
to $\{x:|x_1|=\cdots =|x_d|\}$; also assume that there is an arbitrarily small compact neighborhood $\mathbb N_0$ such that
$\pi(\mathbb N_0)=0$. If there exist compact neighborhoods $\mathcal N_0$, $\mathcal N_{+1}$ and $\mathcal N_{-1}$
(all arbitrarily small) of $0$, $+1$ and $-1$, respectively such that
$g^{(2)}$ gives zero probability to $\mathcal N_0$, $\mathcal N_{+1}$ and $\mathcal N_{-1}$,   
then the multiplicative TMCMC chain satisfies geometric drift if and only if

\begin{equation}
\underset{\|x\| \rightarrow \infty}{\lim\inf}~ Q^{(2)}(x, A^{(2)}(x)) > 0.
\label{eq:liminf_Q_multiplicative}
\end{equation}


\end{theorem}

\begin{proof}

As before, let $C_{\pi(x)}$ be the contour of the density $\pi$ corresponding to the value $\pi(x)$, and
let the radial cone around $ C_{\pi(x)}$ be $C_{\pi(x)} (\delta)$ given by (\ref{eq:radial_cone}).
By (\ref{eq:liminf_Q_multiplicative}) there exists $\gamma > 0$ such that 
\begin{equation}
\underset{\|x\| \rightarrow \infty} {\lim\sup} \hspace{0.2 cm} Q^{(2)} (x, R^{(2)}(x)) \leq 1-5\gamma^{\frac{1}{2}}.
\label{eq:limsup_Q_multiplicative}
\end{equation}

Once again, we take the belt of length $\delta$ such that the probability that a move from $x$ 
falls within this $\delta$ belt is 
less than $\gamma$. This holds since the neighborhoods $\mathcal N_{+1}$ and 
$\mathcal N_{-1}$ of $+1$ and $-1$ receive zero probabilities
under the proposal density $g^{(2)}$.
The remaining arguments are similar as in the proof of Theorem \ref{theorem:geo_additive}. 
Hence, as before,
there exists $R_{\gamma}$ so that for any point $y$ outside the $\delta$ bound around $x$, (\ref{eq:super_exp}) holds. 

As before, let $ V(x) = \frac{c}{\sqrt{\pi(x)}}$, where $c > 0$ is chosen appropriately. 
Then it holds that
\begin{eqnarray}
&&\frac{P^{(2)}V(x)}{V(x)}\notag\\
&= & \frac{1}{3^{d}} \sum_{b_{1},\cdots, b_{d}} 
\int _{A^{(2)}(x)} {\left [ \frac{\pi(x_{1},\ldots,x_{d})}{\pi(T_b(x,\epsilon))} \right ]^{\frac{1}{2}} 
g^{(2)}(\epsilon)d\epsilon} \nonumber\\
&& + \frac{1}{3^{d}} \sum_{b_{1},\cdots, b_{d}} \int _{R^{(2)}(x)}{ 
\left [ 1- \frac{\pi(T_b(x,\epsilon))}{\pi(x_{1},\ldots,x_{d})}|J(b,\epsilon)| 
+ \left \{ \frac{\pi(T_b(x,\epsilon))}{\pi(x_{1},\ldots,x_{d})} \right \}^{\frac{1}{2}}|J(b,\epsilon)| \right ] 
g^{(2)}(\epsilon)d\epsilon}.\notag\\
\label{eq:exact_expression}\\
&\leq & \frac{1}{3^{d}} \sum_{b_{1},\cdots, b_{d}} 
\left\{\sum_{b_1,\ldots,b_d}
\int _{A^{(2)}(b,x)} {\left [ \frac{\pi(x_{1},\ldots,x_{d})}{\pi(T_b(x,\epsilon))} \right ]^{\frac{1}{2}} 
g^{(2)}(\epsilon)d\epsilon}\right\} \nonumber\\
&& + \frac{1}{3^{d}} \sum_{b_{1},\cdots, b_{d}} \int _{R^{(2)}(x)}{ 
\left [ 1- \frac{\pi(T_b(x,\epsilon))}{\pi(x_{1},\ldots,x_{d})}|J(b,\epsilon)| 
+ \left \{ \frac{\pi(T_b(x,\epsilon))}{\pi(x_{1},\ldots,x_{d})} \right \}^{\frac{1}{2}}|J(b,\epsilon)| \right ] 
g^{(2)}(\epsilon)d\epsilon}.\nonumber\\
\label{eq:sum_ratios_multiplicative}
\end{eqnarray}
We now break up the integrals on $A^{(2)}(b,x)$ as sums of the integrals on 
$A^{(2)}(b,x)\cap A^*(b,x)$ and $A^{(2)}(b,x)\cap R^*(b,x)$. Also, we break up the integrals
on $R^{(2)}(x)$ as sums of integrals on $R^{(2)}(x)\cap A^*(b,x)$ and $R^{(2)}(x)\cap R^*(b,x)$.
Since $R^{(2)}(x)=\cap_{b_1,\ldots,b_d}R^{(2)}(b,x)$, these involve the intersections 
$R^{(2)}(b,x)\cap A^*(b,x)$ and $R^{(2)}(b,x)\cap R^*(b,x)$, respectively.

Note that, since $|J(b,\epsilon)|$ is of the form $|\epsilon |^k$, for $k=-d,\ldots,-1,0,1,\ldots,d$, 
and $|\epsilon |\leq 1$ (almost surely), $A^{(2)}(b,x)\cap R^*(b,x)$ is either the null set $\emptyset$
(when $k= -d,-d+1,\ldots,-1,0$), or of the form $A^{(2)}(b,x)\cap R^*(b,x)$$=\{\epsilon: |\epsilon|^k\leq\frac{\pi(T^{(2)}_b(x,\epsilon))}{\pi(x)}<1\}$,
for $k=1,2,\ldots,d$. Hence, for $\|x\|>R_{\gamma}$, by (\ref{eq:super_exp}),
$Q^{(2)}(x,A^{(2)}(b,x)\cap R^*(b,x)\cap [C_{\pi(x)}(\delta)]^c)$
$\leq Q^{(2)}(|\epsilon|^k\leq\frac{\pi(T^{(2)}_b(x,\epsilon))}{\pi(x)}<\gamma)<\gamma/2$,
and for $\delta$ sufficiently small, 
$Q^{(2)}(x,A^{(2)}(b,x)\cap R^*(b,x)\cap C_{\pi(x)}(\delta))<\gamma/2$. Moreover, on $A^{(2)}(b,x)\cap R^*(b,x)\cap C_{\pi(x)}(\delta)$,
$\frac{\pi(x_{1},\ldots,x_{d})}{\pi(T_b(x,\epsilon))}$ is bounded by a finite constant. 
By hypothesis, $\mathcal N_0$ has zero probability
under $g^{(2)}$. 
This implies that the set
\begin{equation*}
\mathcal S=\{|\epsilon |\leq 1:\exists ~b~\mbox{ and set}~ \mathcal S_{\epsilon}~\mbox{such that for}~
x\in \mathcal S_{\epsilon},~ 
\frac{\pi(x_{1},\ldots,x_{d})}{\pi(T_b(x,\epsilon))}>K,~\forall~ K>0\}
\label{eq:S}
\end{equation*}
has zero probability under $g^{(2)}$.
Hence $\frac{\pi(x_{1},\ldots,x_{d})}{\pi(T_b(x,\epsilon))}$ is almost surely bounded even on
$A^{(2)}(b,x)\cap R^*(b,x)\cap C^c_{\pi(x)}(\delta)$.

Hence, by the above arguments, for $\|x\|>R_{\gamma}$,
and for sufficiently small $\xi>0$,
\begin{align}
&\int _{A^{(2)}(b,x)\cap R^*(b,x)} \left [ \frac{\pi(x_{1},\ldots,x_{d})}{\pi(T_b(x,\epsilon))} \right ]^{\frac{1}{2}}
g^{(2)}(\epsilon)d\epsilon\notag\\
&=\int _{A^{(2)}(b,x)\cap R^*(b,x)\cap C_{\pi(x)}(\delta)} \left [ \frac{\pi(x_{1},\ldots,x_{d})}{\pi(T_b(x,\epsilon))} \right ]^{\frac{1}{2}}
g^{(2)}(\epsilon)d\epsilon\notag\\
&+\int _{A^{(2)}(b,x)\cap R^*(b,x)\cap [C_{\pi(x)}(\delta)]^c} \left [ \frac{\pi(x_{1},\ldots,x_{d})}{\pi(T_b(x,\epsilon))} \right ]^{\frac{1}{2}}
g^{(2)}(\epsilon)d\epsilon\notag\\
&<\xi/2\notag.
\end{align}
By similar (in fact, somewhat simpler) arguments, it follows that for $\|x\|>R_{\gamma}$,
\begin{align}
\int _{A^{(2)}(b,x)\cap A^*(b,x)} \left [ \frac{\pi(x_{1},\ldots,x_{d})}{\pi(T_b(x,\epsilon))} \right ]^{\frac{1}{2}}
g^{(2)}(\epsilon)d\epsilon
&<\frac{\xi}{2}+\xi^{\frac{1}{2}}Q^{(2)}(x,A^{(2)}(b,x)\cap A^*(b,x)\cap [C_{\pi(x)}(\delta)]^c)\notag\\
&<\frac{\xi}{2}+\xi^{\frac{1}{2}}Q^{(2)}(x,A^{(2)}(x)).
\label{eq:upper_bound1}
\end{align}
The arguments required are somewhat simpler because on $A^{(2)}(b,x)\cap A^*(b,x)$, the ratio 
$\frac{\pi(x_{1},\ldots,x_{d})}{\pi(T_b(x,\epsilon))}$ is bounded above by 1.
Hence, the first part of the expression for $P^{(2)}V(x)/V(x)$ given by (\ref{eq:sum_ratios_multiplicative}) is
less than $3^d\left(\xi+\xi^{\frac{1}{2}}Q^{(2)}(x,A^{(2)}(x))\right)$.
Formally, for $\|x\|>R_{\gamma}$,
\begin{equation}
\frac{1}{3^{d}} \sum_{b_{1},\cdots, b_{d}} 
\left\{\sum_{b_1,\ldots,b_d}\int _{A^{(2)}(b,x)} {\left [ \frac{\pi(x_{1},\ldots,x_{d})}{\pi(T_b(x,\epsilon))} 
\right ]^{\frac{1}{2}} 
g^{(2)}(\epsilon)d\epsilon}\right\} <3^d\left(\xi+\xi^{\frac{1}{2}}Q^{(2)}(x,A^{(2)}(x))\right).
\label{eq:upper_bound2}
\end{equation}
For sufficiently small $\xi>0$ we can choose $\eta>3^{2d}\xi$ so that
\begin{equation}
3^d\left(\xi+\xi^{\frac{1}{2}}Q^{(2)}(x,A^{(2)}(x))\right)
<\eta+\eta^{\frac{1}{2}}Q^{(2)}(x,A^{(2)}(x)).
\label{eq:new_upper_bound}
\end{equation}
In the second part of the expression for $P^{(2)}V(x)/V(x)$, note that on $R^{(2)}(x)$,  
$$1- \frac{\pi(T_b(x,\epsilon))}{\pi(x_{1},\ldots,x_{d})}|J(b,\epsilon)|<1,$$
so that 
\begin{equation}
\int_{R^{(2)}(x)}\left[1- \frac{\pi(T_b(x,\epsilon))}{\pi(x_{1},\ldots,x_{d})}|J(b,\epsilon)|\right]g^{(2)}(\epsilon)d\epsilon
<Q^{(2)}(x,R^{(2)}(x)),
\label{eq:upper_bound3}
\end{equation}
and
\begin{align}
&\int_{R^{(2)}(x)}\left\{ \frac{\pi(T_b(x,\epsilon))}{\pi(x_{1},\ldots,x_{d})} \right \}^{\frac{1}{2}}|J(b,\epsilon)|  
g^{(2)}(\epsilon)d\epsilon\label{eq:uppbound1}\\
&=\int_{R^{(2)}(x)\cap A^*(b,x)}\left\{ \frac{\pi(T_b(x,\epsilon))|J(b,\epsilon)|}{\pi(x_{1},\ldots,x_{d})} \right \}^{\frac{1}{2}}
|J(b,\epsilon)|^{\frac{1}{2}}  
g^{(2)}(\epsilon)d\epsilon\label{eq:upper_bound4_1st}\\
&+\int_{R^{(2)}(x)\cap R^*(b,x)}\left\{ \frac{\pi(T_b(x,\epsilon))}{\pi(x_{1},\ldots,x_{d})} \right \}^{\frac{1}{2}}|J(b,\epsilon)|  
g^{(2)}(\epsilon)d\epsilon.
\label{eq:upper_bound4}
\end{align}
Note that on $R^{(2)}(x)\cap A^*(b,x)$, $ \frac{\pi(T_b(x,\epsilon))|J(b,\epsilon)|}{\pi(x_{1},\ldots,x_{d})}<1$, and
by our choice of the proposal density $g^{(2)}$, $\mathcal N_0$ has zero probability under $g^{(2)}$,
so that the Jacobians $|J(b,\epsilon)|$ are bounded above by a finite constant, say $K$; we choose
$K>1$. Hence, the first integral (\ref{eq:upper_bound4_1st})
in the break-up of the integral (\ref{eq:uppbound1}) is bounded above by $K Q^{(2)}(x,R^{(2)}(x)\cap A^*(b,x))$.
Now, $Q^{(2)}(x,R^{(2)}(x)\cap A^*(b,x))=Q^{(2)}(x,R^{(2)}(x)\cap A^*(b,x)\cap C_{\pi(x)}(\delta))
+Q^{(2)}(x,R^{(2)}(x)\cap A^*(b,x)\cap [C_{\pi(x)}(\delta)]^c)$, and we can achieve 
$Q^{(2)}(x,R^{(2)}(x)\cap A^*(b,x)\cap C_{\pi(x)}(\delta))<\gamma^*/4$, for sufficiently small $\gamma^*$.

The sets of the form $R^{(2)}(b,x)\cap A^*(b,x)$ are again empty sets or of the form
$\{\epsilon: |\epsilon|^k\leq\frac{\pi(T^{(2)}_b(x,\epsilon))}{\pi(x)}<1\}$; $k=1,2,\ldots,d$. 
Hence, the sets $R^{(2)}(x)\cap A^*(b,x)$ are also either empty sets
or intersections with sets of the form  
$\{\epsilon: |\epsilon|^k\leq\frac{\pi(T^{(2)}_b(x,\epsilon))}{\pi(x)}<1\}$; $k=1,2,\ldots,d$.
Hence, for $\|x\|>R_{\gamma}$, we can achieve $Q^{(2)}(x,R^{(2)}(x)\cap A^*(b,x)\cap [C_{\pi(x)}(\delta)]^c)<\gamma^*/4$.
In other words, for $\|x\|>R_{\gamma}$,
\begin{align}
&\int_{R^{(2)}(x)\cap A^*(b,x)}\left\{ \frac{\pi(T_b(x,\epsilon))|J(b,\epsilon)|}{\pi(x_{1},\ldots,x_{d})} \right \}^{\frac{1}{2}}
|J(b,\epsilon)|^{\frac{1}{2}}  
g^{(2)}(\epsilon)d\epsilon\notag\\
& <K Q^{(2)}(x,R^{(2)}(x)\cap A^*(b,x)\cap C_{\pi(x)}(\delta))
+ Q^{(2)}(x,R^{(2)}(x)\cap A^*(b,x)\cap [C_{\pi(x)}(\delta)]^c)\notag\\
&< K\gamma^*/2.
\label{eq:upp1}
\end{align}

Now consider the second integral (\ref{eq:upper_bound4}) in the break-up of the integral (\ref{eq:uppbound1}). 
We have
\begin{align}
&\int_{R^{(2)}(x)\cap R^*(b,x)}\left\{ \frac{\pi(T_b(x,\epsilon))}{\pi(x_{1},\ldots,x_{d})} \right \}^{\frac{1}{2}}|J(b,\epsilon)|  
g^{(2)}(\epsilon)d\epsilon\notag\\
&=\int_{R^{(2)}(x)\cap R^*(b,x)\cap C_{\pi(x)}}\left\{ \frac{\pi(T_b(x,\epsilon))}{\pi(x_{1},\ldots,x_{d})} \right \}^{\frac{1}{2}}|J(b,\epsilon)|  
g^{(2)}(\epsilon)d\epsilon\label{eq:uppe1}\\
&+\int_{R^{(2)}(x)\cap R^*(b,x)\cap [C_{\pi(x)}]^c}\left\{ \frac{\pi(T_b(x,\epsilon))}{\pi(x_{1},\ldots,x_{d})} \right \}^{\frac{1}{2}}|J(b,\epsilon)|  
g^{(2)}(\epsilon)d\epsilon.\label{eq:uppe2}
\end{align}
Note that on $R^*(b,x)$, $\frac{\pi(T_b(x,\epsilon))}{\pi(x_{1},\ldots,x_{d})}<1$. Hence, the first integral 
(\ref{eq:uppe1}) in the above break-up is bounded above by 
$K Q^{(2)}(x,R^{(2)}(x)\cap R^*(b,x)\cap C_{\pi(x)}(\delta))$, which, in turn, is
bounded above by $K\gamma^*/2$. 
For $\|x\|>R_{\gamma}$, the second integral (\ref{eq:uppe2}) is bounded above by 
$K{\gamma^*}^{\frac{1}{2}}Q^{(2)}(x,R^{(2)}(x)\cap R^*(b,x)\cap [C_{\pi(x)}(\delta)]^c)$, which, in turn, is bounded above by
$K{\gamma^*}^{\frac{1}{2}}Q^{(2)}(x,R^{(2)}(x))$.
In other words,
\begin{align}
&\int_{R^{(2)}(x)\cap R^*(b,x)}\left\{ \frac{\pi(T_b(x,\epsilon))|J(b,\epsilon)|}{\pi(x_{1},\ldots,x_{d})} \right \}^{\frac{1}{2}}
|J(b,\epsilon)|^{\frac{1}{2}}  
g^{(2)}(\epsilon)d\epsilon
< K\frac{\gamma^*}{2}+ K{\gamma^*}^{\frac{1}{2}}Q^{(2)}(x,R^{(2)}(x)).
\label{eq:upp2}
\end{align}
Combining (\ref{eq:upp1}) and (\ref{eq:upp2}) we obtain that (\ref{eq:uppbound1}) is bounded above by
$K\gamma^*+ K{\gamma^*}^{\frac{1}{2}}Q^{(2)}(x,R^{(2)}(x))$.
With sufficiently small $\gamma^*$ we have, for $\eta> K^2\gamma^*$, 
\[
 K\gamma^*+ K{\gamma^*}^{\frac{1}{2}}Q^{(2)}(x,R^{(2)}(x))
<
\eta+\eta^{\frac{1}{2}}Q^{(2)}(x,R^{(2)}(x)).
\]
Combining this with (\ref{eq:upper_bound3}) we get the following upper bound for the second term of
(\ref{eq:sum_ratios_multiplicative}):
\begin{align}
& \frac{1}{3^{d}} \sum_{b_{1},\cdots, b_{d}} \int _{R^{(2)}(x)}{ 
\left [ 1- \frac{\pi(T_b(x,\epsilon))}{\pi(x_{1},\ldots,x_{d})}|J(b,\epsilon)| 
+ \left \{ \frac{\pi(T_b(x,\epsilon))}{\pi(x_{1},\ldots,x_{d})} \right \}^{\frac{1}{2}}|J(b,\epsilon)| \right ] 
g^{(2)}(\epsilon)d\epsilon}\notag\\
&<\eta+\left(1+\eta^{\frac{1}{2}}\right)Q^{(2)}(x,R^{(2)}(x)).
\label{eq:upper_bound5}
\end{align}

Combining (\ref{eq:upper_bound2}) and (\ref{eq:upper_bound5}), we obtain, for $\eta<\gamma$
(so that $\max\{3^{2d}\xi, K^2\gamma^*\}<\eta<\gamma$),
\begin{align}
\underset{\|x\| \rightarrow \infty} {\lim\sup}~
\frac{P^{(2)}V(x)}{V(x)} 
&\leq 2\eta+\eta^{\frac{1}{2}}+\underset{\|x\| \rightarrow \infty} {\lim\sup}~Q^{(2)}(x,R^{(2)}(x))\notag\\
&< 2\eta+\eta^{\frac{1}{2}}+1-5\eta^{\frac{1}{2}}\quad\mbox{(by (\ref{eq:limsup_Q_multiplicative}) and the fact that 
$\eta<\gamma$})\notag\\
&=1-(2\eta)^{\frac{1}{2}}+2\eta\notag\\
&<1.\notag
\end{align}

Hence (\ref{eq:1con}) holds. To see that condition (\ref{eq:infcon}) holds, in (\ref{eq:exact_expression})
observe that 
all the ratios in the integrands are bounded above by 1, while the terms
$|J(b,\epsilon)|^{\frac{1}{2}}$ are almost surely bounded above by our choice of the proposal density $g^{(2)}$. 
Hence, $P^{(2)}V(x)/V(x)$ 
is finite for every $x$.

Now we prove that if multiplicative TMCMC is geometrically ergodic, then (\ref{eq:liminf_Q_multiplicative}) is satisfied. 
As before, we prove that if $\underset{\|x\| \rightarrow \infty} {\lim\sup}~Q^{(2)} (x, R^{(2)}(x))=1$, then 
$\underset{\|x\| \rightarrow \infty} {\lim\sup}~P^{(2)} (x, \{x\})=1$. 
Again, we choose a compact set $E$ such that $Q^{(2)}(x,E^c)\leq\eta$ and choose $\delta>0$
small enough such that 
$\underset{\|x\| \rightarrow \infty}{\lim\sup}~Q^{(2)}(x,C_{\pi(x)}(\delta)\cap E)\leq\eta$.
Since $R^{(2)}(x)\cap A^*(b,x)\cap [C_{\pi(x)}(\delta)]^c$ is either null set or intersection with
sets of the form
$\{\epsilon:|\epsilon|^k\leq \frac{\pi(x)}{\pi(T_b(x,\epsilon))}<1\}$, for $k=1,2,\ldots,d$, it follows from (\ref{eq:super_exp})
that for any fixed $b^*$, if $\|x\|>R_{\eta}$, 
\[
Q^{(2)}(x,R^{(2)}(x)\cap A^*(b^*,x)\cap [C_{\pi(x)}(\delta)]^c)
\leq Q^{(2)}(x,\{\epsilon:|\epsilon|^k\leq \frac{\pi(x)}{\pi(T_{b^*}(x,\epsilon))}<\eta\})\leq\eta.
\]
Hence, 
\[
\underset{\|x\| \rightarrow \infty}{\lim\sup}~Q^{(2)}
(x, R^{(2)}(x)\cap A^*(b^*,x)\cap [C_{\pi(x)}(\delta)]^c)\leq\eta.
\]
Since 
$\underset{\|x\| \rightarrow \infty} {\lim\sup}~Q^{(2)} (x, R^{(2)}(x))=1$, the above imply 
that 
\[
\underset{\|x\| \rightarrow \infty}{\lim\sup}~Q^{(2)}
(x, R^{(2)}(x)\cap R^*(b^*,x)\cap [C_{\pi(x)}(\delta)]^c)>1-2\eta.
\]
Moreover, since $|J(b,\epsilon)|$ are almost surely bounded by the choice of our proposal density,
assume that there exists $0< K<\infty$ such that 
$|J(b,\epsilon)|< K$ almost surely with respect to $g^{(2)}$. 

These and the fact (\ref{eq:super_exp}) 
imply that 
\begin{eqnarray}
\underset{\|x\| \rightarrow \infty}{\lim\sup}~P^{(1)}(x,\{x\}) 
&=& \underset{\|x\| \rightarrow \infty}{\lim\sup}~\frac{1}{3^{d}} 
\sum_{b_{1},\cdots, b_{d}}\int _{R^{(2)}(x)}{ 
\left [ 1- \frac{\pi(T_b(x,\epsilon))}{\pi(x)}|J(b,\epsilon)| \right ] g^{(2)}(\epsilon)d\epsilon}\nonumber\\
& \geq & 
\underset{\|x\| \rightarrow \infty}{\lim\sup}~
\frac{1}{3^{d}} \sum_{b_{1},\cdots, b_{d}} \int _{R^{(2)}(x)\cap R^*(b^*,x)\cap [C_{\pi(x)}(\delta)]^c}{ 
\left [ 1- \frac{\pi(T_b(x,\epsilon))}{\pi(x)}|J(b,\epsilon)| \right ] g^{(2)}(\epsilon)d\epsilon}\nonumber\\ 
&\geq & (1-\eta K)
\underset{\|x\| \rightarrow \infty}{\lim\sup}~
Q(x,R^{(2)}(x)\cap R^*(b^*,x)\cap [C_{\pi(x)}(\delta)]^c)\nonumber\\
&\geq & (1-\eta K)(1-2\eta).
\label{eq:non_geo_multiplicative}
\end{eqnarray}
Since $\eta>0$ is arbitrary, the proof is complete.

\end{proof}
That it is easy to ensure geometric ergodicity of multiplicative TMCMC in super-exponential cases
can be seen as follows. Select a move-type $b^*$ such that $|J(b^*,\epsilon)|=|\epsilon |$.
Then $A(b^*,x)=\{\epsilon:\frac{\pi(T_{b^*}(x,\epsilon))}{\pi(x)}|\epsilon |\geq 1\}$, and
$A^*(b^*,x)=\{\epsilon:\frac{\pi(T_{b^*}(x,\epsilon))}{\pi(x)}\geq 1\}$. Then, since $|\epsilon |\leq 1$
almost surely,
\begin{equation}
A(b^*,x)\cap A^*(b^*,x)=\left\{\epsilon:\frac{\pi(x)}{\pi(T_{b^*}(x,\epsilon))}<|\epsilon |\leq 1\right\}.
\label{eq:eqn1}
\end{equation}
If, for $\eta>0$, $\|x\|>R_{\eta}$, then by (\ref{eq:super_exp}), 
\begin{equation}
\frac{\pi(x)}{\pi(T_{b^*}(x,\epsilon))}<\eta.
\label{eq:eqn2}
\end{equation}
Equations (\ref{eq:eqn1}) and (\ref{eq:eqn2}) imply that for any given $\xi>0$ it is possible to choose $\eta>0$ 
such that for $\|x\|>R_{\eta}$, it holds that
\begin{equation}
Q^{(2)}(x,A^{(2)}(b^*,x)\cap A^*(b^*,x))>1-\xi.
\label{eq:eqn3}
\end{equation}
Hence, for $\|x\|>R_{\eta}$, we obtain using (\ref{eq:eqn3}),
\begin{align}
Q^{(2)}(x,A^{(2)}(x))&\geq Q^{(2)}(x,A^{(2)}(x)\cap A^*(b^*,x))\notag\\
&\geq Q^{(2)}(x,A^{(2)}(b^*,x)\cap A^*(b^*,x))\notag\\
&>1-\xi.\notag
\end{align}
Hence, (\ref{eq:liminf_Q_multiplicative}) holds, ensuring geometric ergodicity.

As we remarked earlier, we omit the proof of geometric ergodicity of additive-multiplicative TMCMC, since it
is almost the same as that of multiplicative TMCMC, provided above.

\section{Illustration with simulation studies}
\label{sec:simulation}

There are several considerations in defining the accuracy or the efficiency of any MCMC-based approach. 
First, one important aspect is that the chain must have reasonably high acceptance rate. 
This has been an important consideration in our proposing TMCMC. 
It is to be noted that geometric ergodicity only tells us that convergence 
of our chain to the target density occurs at a geometric rate. 
However, if the value of $\rho$, the geometric rate in (\ref{eq:geo}) is close to 1, then the 
algorithm in question, 
in spite of being geometrically ergodic, need not be efficient in practice. 
To test how efficient our TMCMC algorithms actually are in absolute terms and also relative to standard MCMC 
approaches, we need to define a measure of closeness of the $n$-th order kernel $P^{n}(x,\cdot)$ 
with respect to the target density $\pi(\cdot)$, assuming that the latter can be empirically evaluated. 
The Kolmogorov Smirnov (K-S) distance seems to be a suitable candidate in this regard, and the one
that we adopt for our purpose. Corresponding to each MCMC algorithm, we consider $N$ replicates of the 
chain starting from the same initial value, so that at each iteration $t$, we obtain a set of
$N$ many realizations of the chain. We then compute the empirical distribution of these $N$ values 
and measure the K-S distance between the empirical distribution and the target distribution $\pi$. 
For the chain to be efficient, it must have K-S distance close to 0 after the chain has run for a 
large number of iterations (that is, when $t$ is large). Moreover, the burn-in period is expected to be small
for efficient MCMC algorithms.

\subsection{First simulation experiment comparing RWMH and additive TMCMC}
\label{subsec:sim1}

Table \ref{table:table1} presents the results of a simulation experiment comparing
the performances of RWMH and additive TMCMC (Add-TMCMC) chains for different dimensions,
where, for our purpose we consider the target density $\pi$ to be the multivariate
normal distribution with mean vector $\bzero$ and covariance matrix $\bold{I}$, the identity matrix.
For RWMH we consider two distinct scales for the normal random walk proposal for each
of the coordinates -- the optimal scale 2.4, and a sub-optimal scale 6. We consider the same scaling for
additive TMCMC as well.
Indeed, as shown in \ctn{Dey13}, for both additive TMCMC and RWMH, the optimal scaling parameter 
is very close to 2.4 but the optimal acceptance rate of additive TMCMC is around 0.439, which is 
significantly higher than 0.234, the optimal acceptance rate of the RWMH approach 
(\ctn{Roberts96}, \ctn{Roberts1997}). Moreover, the results of simulation experiments
reported in \ctn{Dey13} demonstrate superior performance of additive TMCMC 
over RWMH in terms of higher acceptance rates irrespective of dimensions and optimal or sub-optimal
scale choices.

Referring to Table \ref{table:table1}, since the K-S statistic is computed after burn in, 
the differences between additive TMCMC and RWMH in terms 
of the K-S distance do not appear to be pronounced in low dimensions, but for dimensions 100 and 200, 
the differences seem to be more pronounced, indicating somewhat better performance of TMCMC.


Figure \ref{fig:figex1_30} displays the
K-S distances corresponding to RWMH and additive TMCMC when the target is a 30-dimensional normal distribution.
It is clearly seen that additive TMCMC converges much faster than RWMH. In fact, the figures indicate that
additive TMCMC takes around just 150 iterations to converge when the scale is optimal, and around 200
iterations when the scale is sub-optimal. On the other hand, in the case of optimal scaling, RWMH takes
around 300 iterations to converge and for sub-optimal scaling it takes around 450 iterations.
The mixing issue is quite pronounced in higher dimensions. Indeed, as seen
in Figure \ref{fig:figex1}, 
the K-S distances associated with additive TMCMC are almost uniformly smaller than those associated with
RWMH, particularly when the scaling is sub-optimal. In fact, in the sub-optimal case it seems that
additive TMCMC has converged within the first 2,000 iterations, whereas RWMH does not seem to show any sign of convergence
even after 20,000 iterations (the K-S distances are significantly larger than those of additive TMCMC).

\begin{table}[h]
\centering
\caption{ 
Performance evaluation of RWMH and additive TMCMC (Add-TMCMC) chains for different dimensions. 
}
\vspace{0.3 in}
\begin{tabular}{|p{0.4in}|c|c|c|c|c|}
\hline
\multirow{2}{*}{Dim.} & \multirow{2}{*}{\backslashbox{Scaling}{Criteria}} & \multicolumn{2}{|c|}{$\begin{array}{c} Acceptance \\ rate ($\%$) \end{array} $}  & \multicolumn{2}{|c|}{\emph{Avg. K-S dist.}} \\ \cline{3-6}
& & RWMH & Add-TMCMC & RWMH & Add-TMCMC\\ \hline

\multirow{3}{*}{2} & 2.4 & 34.9 & 44.6  & 0.1651 & 0.1657 \\ 
 & 6 & 18.66 & 29.15 & 0.1659 & 0.1655 \\ 
\multirow{3}{*}{5} & 2.4 (opt) & 28.6 & 44.12  & 0.1659 & 0.1664 \\ 
& 6 & 2.77 & 20.20  & 0.1693 & 0.1674 \\ 
\multirow{3}{*}{10} & 2.4 (opt) & 26.05 & 44.18  & 0.1652 & 0.1677 \\ 
& 6 & 1.19 & 20.34 & 0.1784 & 0.1688 \\ 
\multirow{2}{*}{100} & 2.4 (opt) & 23.3 & 44.1  & 0.1594 & 0.1571 \\ 
& 6 & 0.32 & 20.6  & 0.1687 & 0.1645 \\ 
\multirow{2}{*}{200} & 2.4 (opt) & 23.4 & 44.2  & 0.1596 & 0.1435 \\ 
& 6 & 0.38 & 20.7   &  0.1622 & 0.1484 \\ \hline
\end{tabular} 
\label{table:table1}
\end{table}

\begin{figure}
\centering
\subfigure [RWMH vs Add-TMCMC (scale = 2.4)]{ \label{fig:K-S1_30}
\includegraphics[trim= 0cm 9cm 0cm 9cm, clip=true, width=12cm,height=6cm]{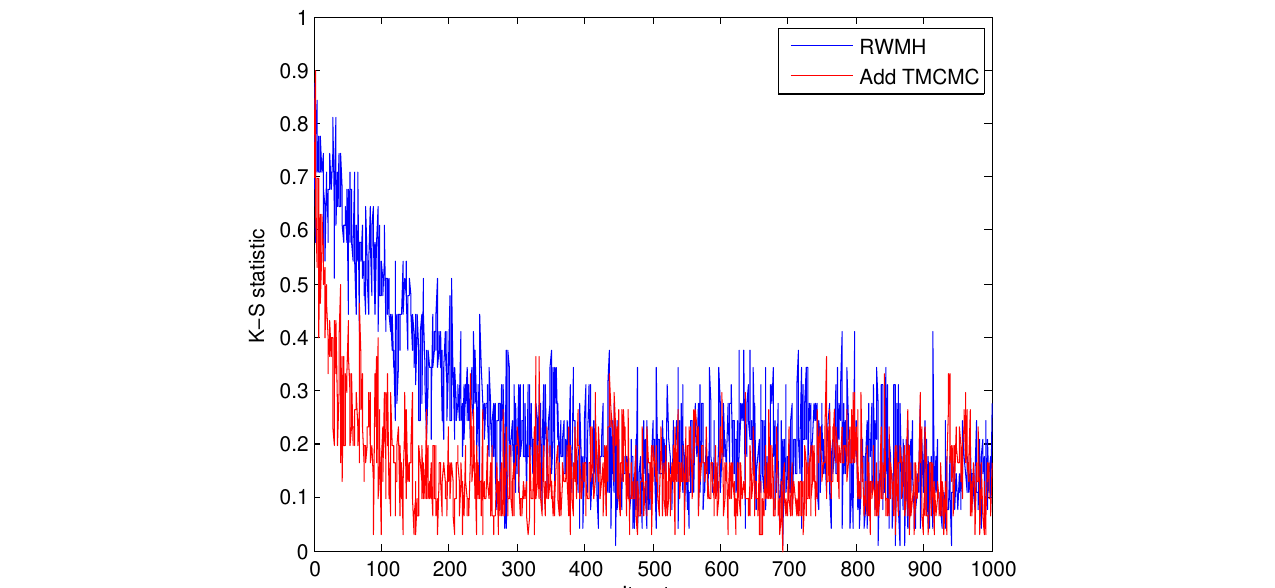}}\\
\subfigure [RWMH vs Add-TMCMC (scale = 6)]{ \label{fig:K-S2_30}
\includegraphics[trim= 0cm 9cm 0cm 9cm, clip=true, width=12cm,height=6cm]{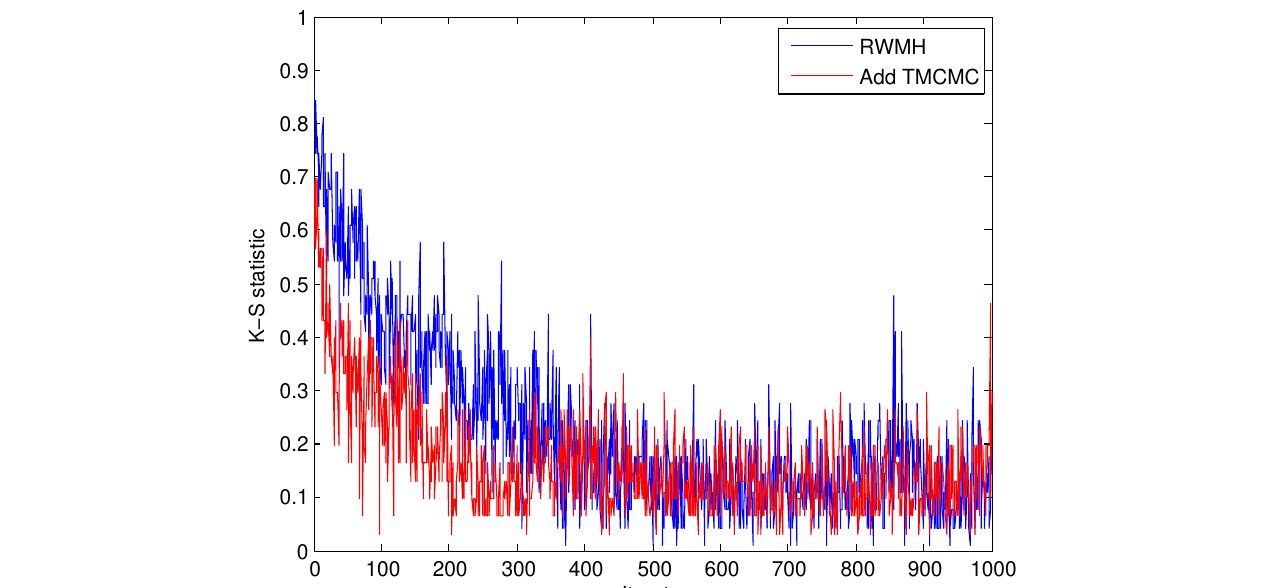}}
\caption{Comparisons between K-S distances associated with additive TMCMC and RWMH for dimension = 30.}
\label{fig:figex1_30}
\end{figure}

\begin{figure}
\centering
\subfigure [RWMH vs Add-TMCMC (scale = 2.4)]{ \label{fig:K-S1}
\includegraphics[trim= 0cm 9cm 0cm 9cm, clip=true, width=12cm,height=6cm]{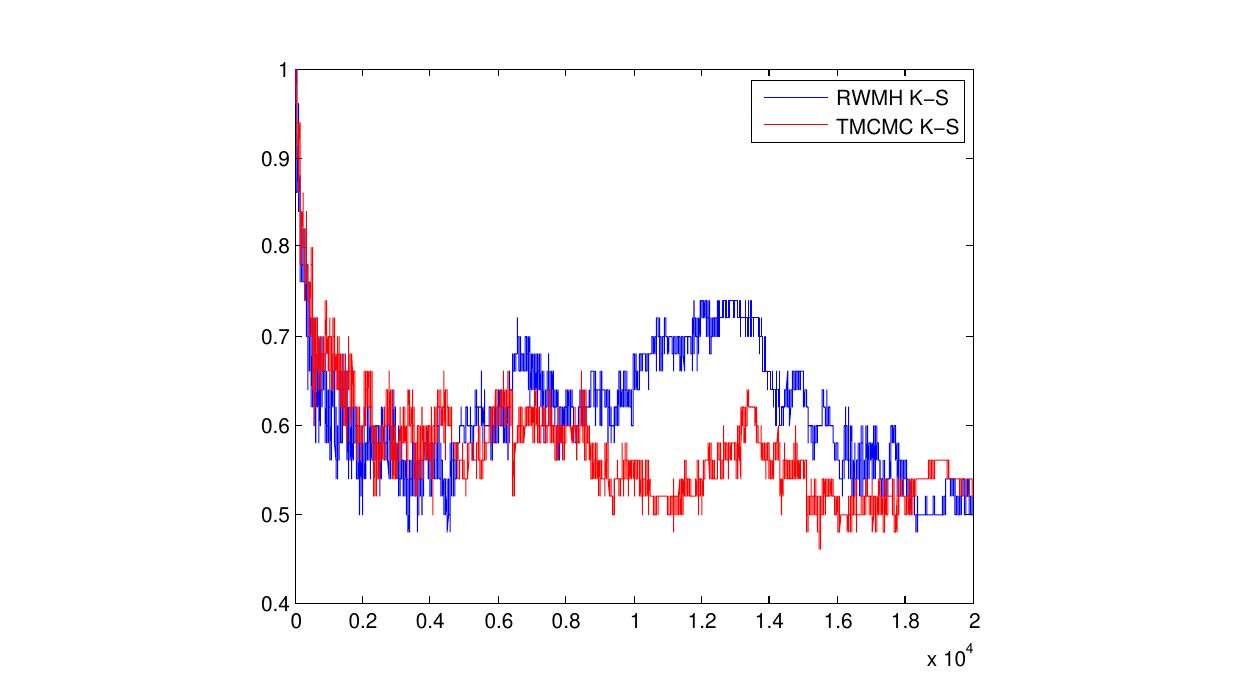}}\\
\subfigure [RWMH vs Add-TMCMC (scale = 6)]{ \label{fig:K-S2}
\includegraphics[trim= 0cm 9cm 0cm 9cm, clip=true, width=12cm,height=6cm]{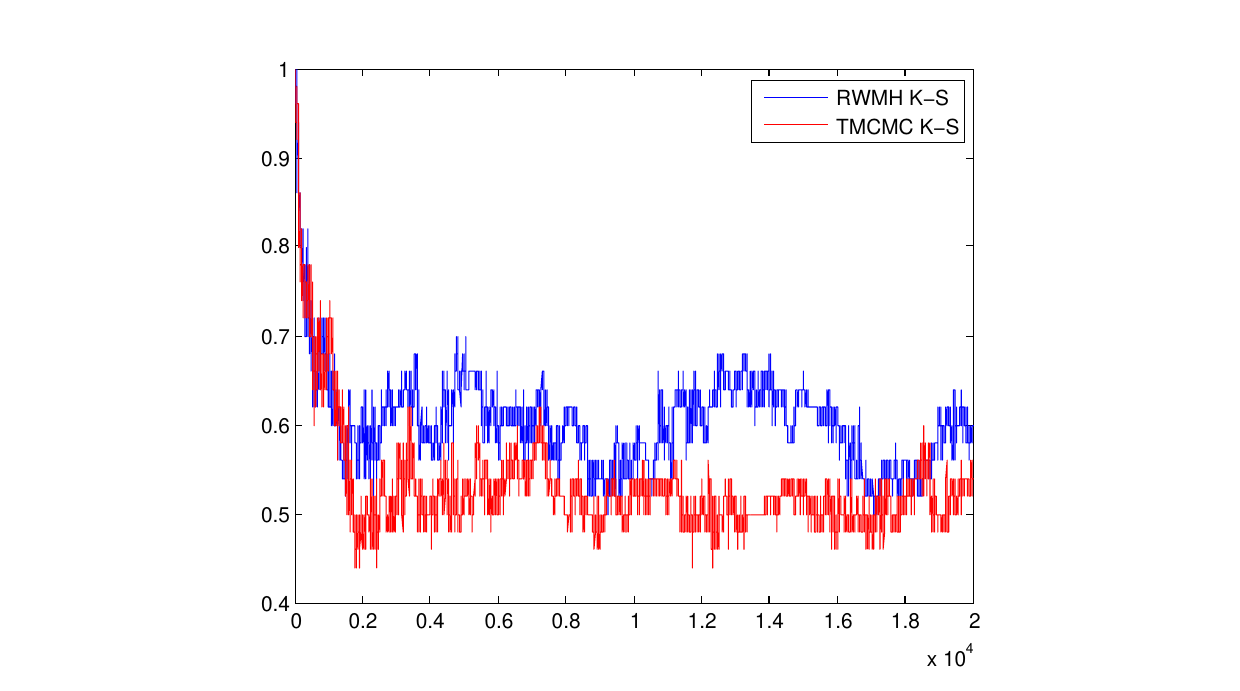}}
\caption{Comparisons between K-S distances associated with additive TMCMC and RWMH for dimension = 100.}
\label{fig:figex1}
\end{figure}

\subsection{Performance comparison with ``essentially fully" multiplicative TMCMC}
\label{subsec:sim2}

In this case, we choose our neighborhood $ \mathbb{N}_{0}$ in (\ref{eq:mc1}) to be $[-0.1,0.1]^{d}$, 
where $d$ is the dimension of the space. The method of  estimation of the mixing probability 
$\pi(\mathbb {N}_{0})$ is discussed in detail in Appendix \ref{sec:implement_P}; however, in our simulation example,
this probability is simply $\prod_{i=1}^d\left[2\Phi(0.1)-1\right]$, $\Phi$ denoting the cumulative distribution function
of $N(0,1)$.
This chain is basically as close as we can get to a fully multiplicative TMCMC chain on 
$\mathbb{R}^{d}$ ensuring that the geometric drift condition holds. 

For our experiment, the scale of the additive TMCMC part of the mixture remains the same
as before, that is, we consider the optimal scale 2.4, and the sub-optimal scale 6. 
We assume the proposal density $g^{(2)}$ is defined on a set of the form $[-l_{2},-l_{1}] \cup [l_{1},l_{2}] $ such that
the interval $[l_{1},l_{2}]$ is a proper subset of $[0,1]$ minus small neighborhoods of 0 and 1. 
The distribution of the step $\epsilon$ is taken to be a mixture normal random variable such that 
$\epsilon\sim \frac{1}{2} N(\mu,\sigma^{2})I_{[l_{1},l_{2}]} +  \frac{1}{2} N(-\mu,\sigma^{2})I_{[-l_{2},-l_{1}]}$ 
with mean 
$\mu \in [l_{1},l_{2}]$ and variance $\sigma^2$.  
In our simulation experiment we assumed $l_{1} =0.05$ and $l_{2}=0.95$ and optimal performance 
was observed when the mean $\mu$ is in the range $0.35$ to $0.45$, which is around halfway from 
both $l_{1}$ and $l_{2}$. 

Table \ref{table:table2} provides a comparison of the performances between
RWMH and essentially full multiplicative TMCMC with respect to acceptance rate
and average K-S distance.
Note that unlike additive TMCMC, we find here that 
the acceptance rate for essentially fully multiplicative TMCMC is poor compared to RWMH. Moreover, 
the K-S distances also suggest that 
RWMH is closer to the target distribution 
compared to essentially fully multiplicative TMCMC for most of the iterations considered. However, on inspection it is observed that the 
K-S distance initially drops faster for the latter compared to RWMH; see Figure \ref{fig:figex2}. 
As shown by \ctn{Dutta12}, multiplicative TMCMC in one-dimensional situations are appropriate
for certain heavy-tailed distributions. But in our current simulation study associated
with high dimensions and a thin-tailed density, (essentially fully) multiplicative TMCMC
did not seem to perform satisfactorily, although theoretically it is geometrically ergodic.

\begin{table}[h]
\centering
\caption{Performance evaluation of RWMH and essentially fully multiplicative TMCMC (Mult-TMCMC) chains 
for different dimensions.} 
\vspace{0.3 in}
\begin{tabular}{|p{0.4in}|c|c|c|c|c|}
\hline
\multirow{2}{*}{Dim} & \multirow{2}{*}{\backslashbox{Scaling}{Criteria}} & \multicolumn{2}{|c|}{$\begin{array}{c} Acceptance \\ rate ($\%$) \end{array} $}  & \multicolumn{2}{|c|}{\emph{Avg. K-S dist.}} \\ \cline{3-6}
& & RWMH & Mult-TMCMC & RWMH & Mult-TMCMC\\ \hline

\multirow{3}{*}{10} & 2.4 (opt) & 26.05 & 16.86 & 0.1652 & 0.2097 \\ 
& 6 & 1.19  & 6.32 & 0.1784 & 0.2133\\ 
\multirow{2}{*}{30} & 2.4 (opt) & 23.5 & 15.74 & 0.1637 & 0.1828 \\ 
& 6 & 1.16 & 6.77 & 0.1711 & 0.1924\\ 
\multirow{2}{*}{100} & 2.4 (opt) & 23.4 & 15.46 & 0.1596 & 0.1812\\ 
& 6 & 0.38 & 2.67 & 0.1622 & 0.1866 \\ \hline
\end{tabular} 
\label{table:table2}
\end{table}

\begin{figure}
\centering
\subfigure [RWMH vs Mult-TMCMC (RWMH scale = 2.4)]{ \label{fig:K-S1_mix}
\includegraphics[trim= 0cm 9cm 0cm 9cm, clip=true, width=12cm,height=6cm]{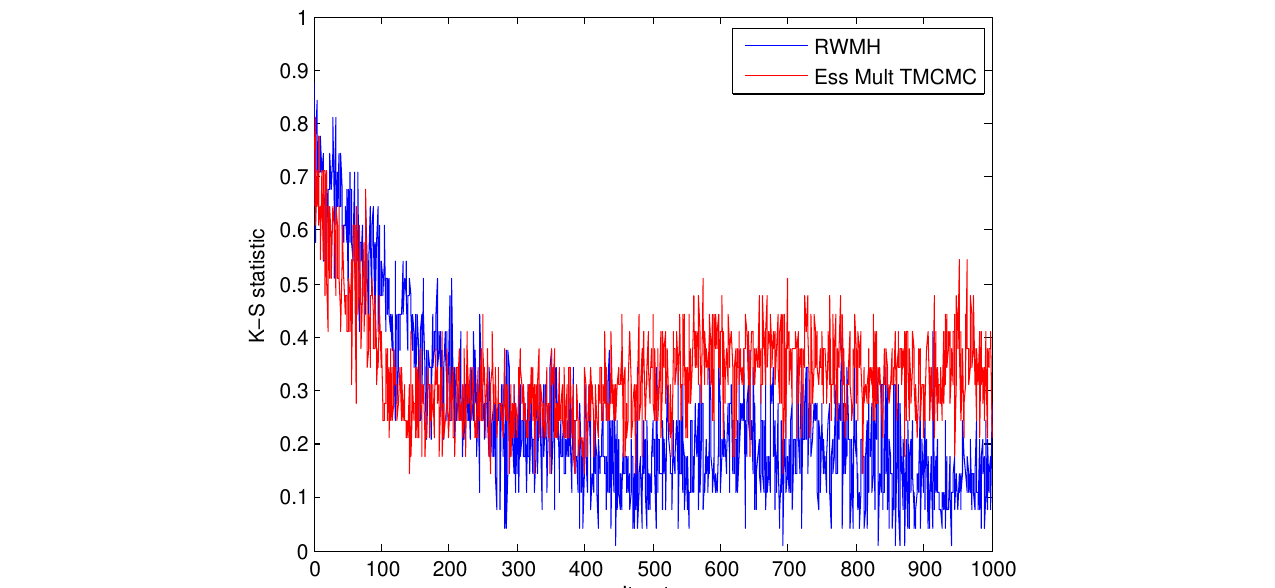}}\\
\subfigure [ RWMH vs Mult-TMCMC (RWMH scale = 6)]{ \label{fig:K-S2_mix}
\includegraphics[trim= 0cm 9cm 0cm 9cm, clip=true, width=12cm,height=6cm]{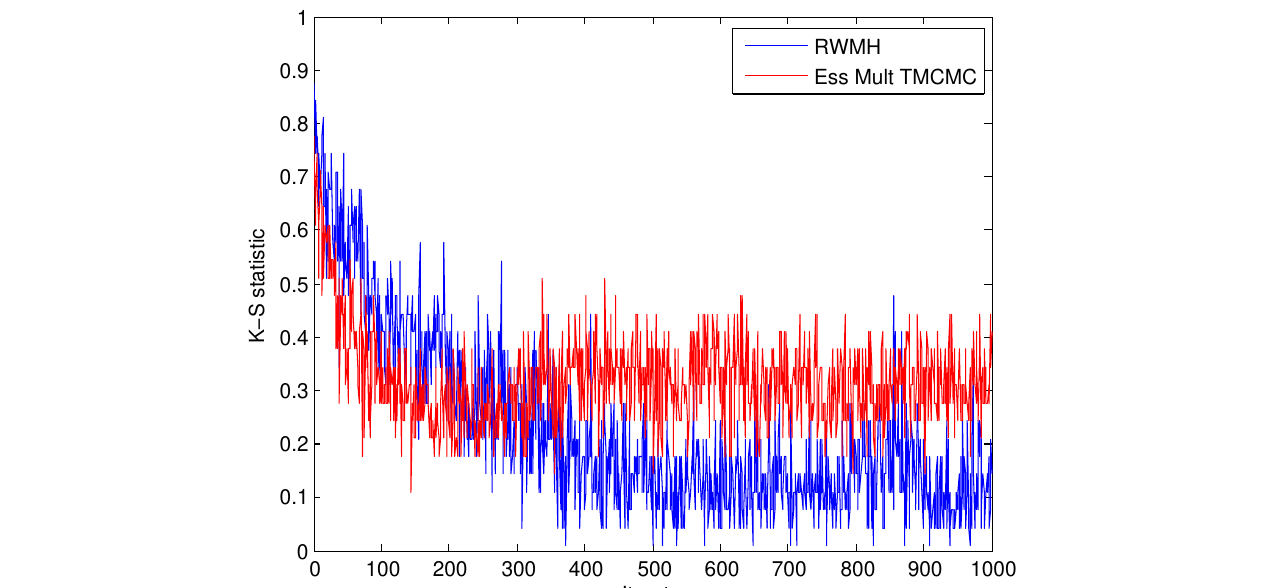}}
\caption{Comparisons between K-S distances associated with essentially Mult-TMCMC and RWMH for dimension = 30.}
\label{fig:figex2}
\end{figure}

\subsection{Performance comparison with the traditional mixture of additive
and multiplicative TMCMC}
\label{subsec:sim3}

Now we consider the traditional mixture chain of the form (\ref{eq:mc2}) 
with both additive and multiplicative moves. We assume that 
with probability $\frac{1}{2}$, we move by additive TMCMC and with probability $\frac{1}{2}$ 
by multiplicative TMCMC. The proposal mechanisms for additive and multiplicative TMCMC
remain the same as in Section \ref{subsec:sim2} associated with essentially fully multiplicative TMCMC.
%
%

\begin{table}[h]
\centering
\caption{Performance evaluation of RWMH and traditional Mixture TMCMC (Mix-TMCMC) 
chains for different dimensions. For the multiplicative TMCMC part, we consider
$\mu=0.35$ and $\sigma=1$.} 
\vspace{0.2 in}
\begin{tabular}{|p{0.4in}|c|c|c|c|c|}
\hline

\multirow{2}{*}{Dim} & \multirow{2}{*}{\backslashbox{Scaling}{Criteria}} & \multicolumn{2}{|c|}{$\begin{array}{c} Acceptance \\ rate ($\%$) \end{array} $}  & \multicolumn{2}{|c|}{\emph{Avg. K-S dist.}} \\ \cline{3-6}
& & RWMH & Mix TMCMC & RWMH & Mix TMCMC\\ \hline

\multirow{3}{*}{10} & 2.4 (opt) & 26.05 & 29.43 & 0.1652 & 0.1455 \\ 
& 6 & 1.19  & 11.26 & 0.1784 & 0.1576\\ 
\multirow{2}{*}{30} & 2.4 (opt) & 23.5 & 29.32 & 0.1637 & 0.1428 \\ 
& 6 & 1.16 & 16.33 & 0.1711 & 0.1529 \\ 
\multirow{2}{*}{100} & 2.4 (opt) & 23.4 & 29. 29 & 0.1596 & 0.1398\\ 
& 6 & 0.38 & 10.67 & 0.1622 & 0.1412 \\ \hline
\end{tabular} 
\label{table:table3}
\end{table}

Table \ref{table:table3} provides a comparison of the performances between
RWMH and our traditional mixture TMCMC kernel with respect to acceptance rate
and average K-S distance.
Note that although the acceptance rate for the mixture kernel in our experiments is around 0.293 for 
$\mu=0.35$ and $\sigma=1$ which is quite low compared to additive TMCMC, it is of course still
significantly higher than the optimal acceptance rate 0.234 for standard RWMH. To avoid any possible
confusion it is important to emphasize that this 
acceptance rate for mixture kernel is not analytically derived as the optimal acceptance rate, 
rather it is the rate corresponding to the optimal value of $\mu$, numerically obtained by varying over 
$\mu$ keeping $\sigma$ fixed at 1 and computing the K-S distance and then choosing that $\mu$ for which 
the empirical average K-S distance was found to be the minimum. However, the average K-S distance for 
the mixture kernel is smaller compared to both RWMH and additive TMCMC, implying faster convergence.
This improvement acts as a trade off for the low acceptance rate of the mixture kernel.

Figure \ref{fig:figex3} displays plots of K-S distances associated with RWMH and mixture TMCMC
in the case of a 30-dimensional normal target distribution. 
The plot shows much faster convergence of mixture TMCMC compared to RWMH.
From Figures \ref{fig:figex1_30} and \ref{fig:figex2}, it is also clear that mixture TMCMC converges
faster than even additive TMCMC and essentially fully multiplicative TMCMC.
In fact, mixture TMCMC seems to converge in just about 100 iterations.
This faster convergence may be attributed to the fact that the multiplicative steps allow the
chain to take longer jumps and hence explore the space faster, while on the other hand the additive
steps keep the acceptance rate high and enables the chain to move briskly. So, in other words,
mixture TMCMC shares the positives of both the additive and the multiplicative chains and is found
to outperform each of them individually.

\begin{figure}
\centering
\subfigure [RWMH vs Mix-TMCMC (scale = 2.4)]{ \label{fig:K-S1_mult}
\includegraphics[trim= 0cm 9cm 0cm 9cm, clip=true, width=12cm,height=6cm]{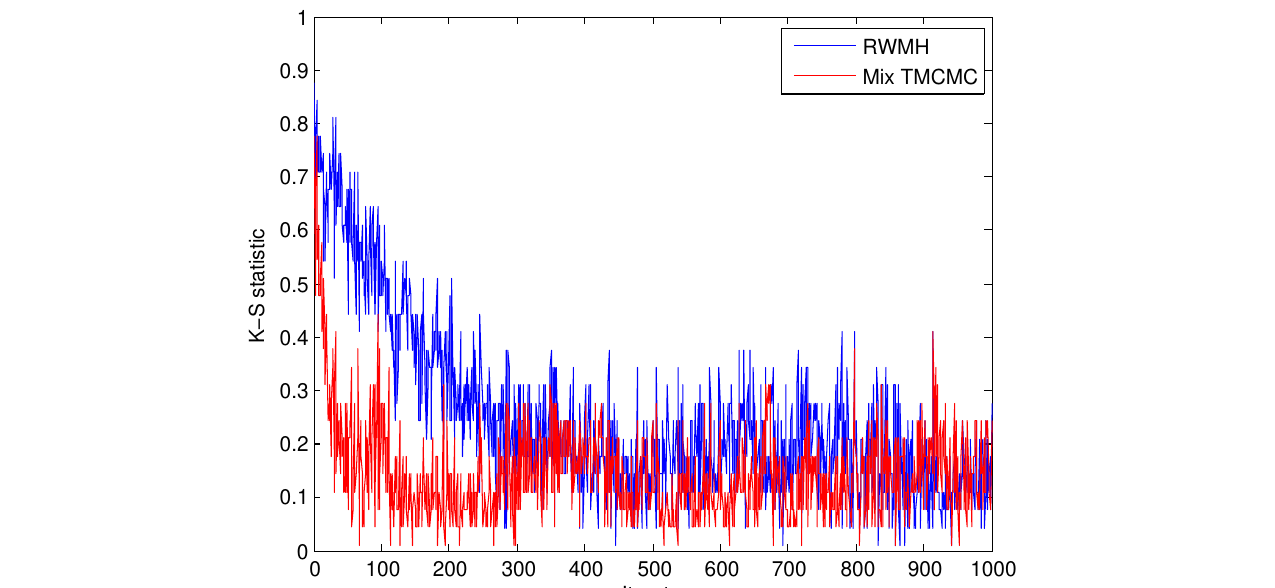}}\\
\subfigure [ RWMH vs Mix-TMCMC (scale = 6)]{ \label{fig:K-S2_mult}
\includegraphics[trim= 0cm 9cm 0cm 9cm, clip=true, width=12cm,height=6cm]{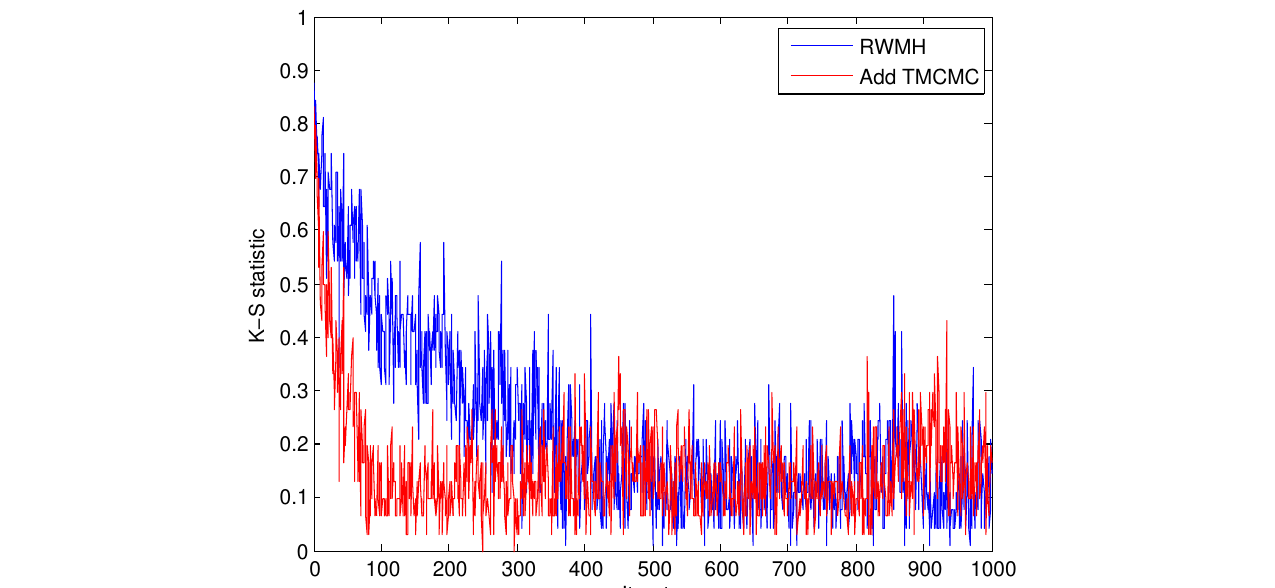}}
\caption{Comparisons between K-S distances associated with Mix-TMCMC and RWMH for dimension = 30.}
\label{fig:figex3}
\end{figure}

\section{Extensions of our geometric ergodicity results to target distributions that are not super-exponential}
\label{sec:non_super_exponential}

So far we have proved geometric ergodicity of additive and multiplicative TMCMC when the target density
$\pi$ is super-exponential. It is natural to ask if our results go through when the 
super-exponential assumption does not hold. 

\subsection{Target density as mixture}
\label{subsec:target_mixture}

Note that, if the target density $\pi$ can be represented as a mixture of the form
\begin{align}
\pi(x)&=\int f_1(x|\theta)f_2(\theta)d\theta,
\label{eq:non_super_exp_pi}
\end{align}
where $f_1(\cdot\vert\theta)$ is super-exponential for all $\theta$ and $f_2$ admits direct (exact) simulation,
then the Markov transition kernel 
\begin{align}
P(x,A)&=\int P(x,A|\theta)f_2(\theta)d\theta,
\label{eq:non_super_exp_P}
\end{align}
where $P(x,A|\theta)$ denotes either additive or multiplicative TMCMC based Markov transtition kernel conditional
on $\theta$, is geometrically ergodic for the target density $\pi$. The proof is essentially the same
as the proof presented in Appendix \ref{sec:P_geo} that the finite mixture Markov 
transition kernel (\ref{eq:mc1}) is geometrically ergodic for the
mixture representation (\ref{eq:mixture}); only the summations need to be replaced with integrals.
The kernel (\ref{eq:non_super_exp_P}) will be implemented by first directly simulating 
$\theta\sim f_2$; then given $\theta$, the transition mechanism $P(x,\cdot|\theta)$ has to
be implemented.

Two popular examples of multivariate densities admitting mixture forms are multivariate $t$
and multivariate Cauchy, both of which can be represented as univariate $Gamma$-distributed 
mixtures of multivariate normal distributions.

\subsection{Change-of-variable idea}
\label{subsec:change_of_variable}

The general situation has been addressed by \ctn{Johnson12} using a change-of-variable idea.
If $\pi_{\beta}$ is the multivariate target density of interest, then one can first simulate
a Markov chain having invariant density
\begin{align}
\pi_{\gamma}(\gamma)=\pi_{\beta}\left(h(\gamma)\right)\left|\mbox{det}~\nabla h(\gamma)\right|,
\label{eq:transformed_target}
\end{align}
where $h$ is a diffeomorphism. If $\pi_{\beta}$ is the density of the random vector $\beta$,
then $\pi_{\gamma}$ is the density of the random vector $\gamma=h^{-1}(\beta)$.
\ctn{Johnson12} obtain conditions on $h$ which make $\pi_{\gamma}$ super-exponentially light.
In more details, \ctn{Johnson12} define the following isotropic function $h:\mathbb R^d\mapsto\mathbb R^d$:
\begin{equation}
h(\gamma)=\left\{\begin{array}{cc}f(\|\gamma\|)\frac{\gamma}{\|\gamma\|}, & \gamma\neq \bzero\\
0 & \gamma=\bzero
\end{array}\right.
\label{eq:isotropy}
\end{equation}
for some function $f: (0,\infty)\mapsto (0,\infty)$.
\ctn{Johnson12} confine attention to isotropic diffeomorphisms, that is, to functions $h$ where 
both $h$ and $h^{-1}$ are continuously differentiable, with the further property that 
$\mbox{det}~\nabla h$ and  $\mbox{det}~\nabla h^{-1}$ are also continuously 
differentiable. In particular, they define $f:[0,\infty)\mapsto [0,\infty)$ as follows:
\begin{equation}
f(x)=\left\{\begin{array}{cc}x, & x<R\\
x+(x-R)^p, & x\geq R,
\end{array}\right.
\label{eq:diffeo}
\end{equation}
where $R\geq 0$ and $p>2$.

Theorem 2 of \ctn{Johnson12} shows that if $\pi_{\beta}$ is an exponentially light density
($\pi_{\beta}$ is exponentially light if 
$\underset{\|x\|\rightarrow\infty}{\lim\sup}~n(x)'\nabla\log\pi_{\beta}(x)<0$)
on $\mathbb R^d$, and
$h$ is defined by (\ref{eq:isotropy}) and (\ref{eq:diffeo}) then the transformed density $\pi_{\gamma}$
given by (\ref{eq:transformed_target}) is super-exponentially light.
Thus, this transformation transforms an exponential density to a super-exponential density.
Theorem 3 of \ctn{Johnson12} provided conditions under which sub-exponential densities
can be converted to exponential densities
($\pi_{\beta}$ is sub-exponentially light if 
$\underset{\|x\|\rightarrow\infty}{\lim\sup}~n(x)'\nabla\log\pi_{\beta}(x)=0$).
In particular, if $\pi_{\beta}$ is a sub-exponentially light density on $\mathbb R^d$, there exist
$\alpha>d$, $R<\infty$ such that
\[
\left(\frac{\beta}{\|\beta\|}\right)'\nabla\log\pi_{\beta}(\beta)\leq -\frac{\alpha}{\|\beta\|},\hspace{2mm}
\|\beta\|>R,
\]
then $h$ defined as (\ref{eq:isotropy}) with $f:[0,\infty)\mapsto [0,\infty)$ given by 
\begin{equation}
f(x)=\left\{\begin{array}{cc}e^{bx}-\frac{e}{3}, & x>\frac{1}{b}\\
x^3\frac{b^3e}{6}+x\frac{be}{2}, & x\leq \frac{1}{b},
\end{array}\right.
\label{eq:diffeo2}
\end{equation}
where $b>0$, ensures that the transformed density $\pi_{\gamma}$ of the form (\ref{eq:transformed_target}),
is super-exponentially light.

In other words, starting from a sub-exponential target density, one can achieve a super-exponential
density by first converting it to exponential using the transformation $h$ (given by (\ref{eq:isotropy})) 
with $f$ given by (\ref{eq:diffeo2}).
Then one can convert the obtained exponential density to super-exponential using the transformation $h$
and $f$ given by (\ref{eq:diffeo}).
As an example \ctn{Johnson12} show that the multivariate $t$ distribution of the form
\begin{equation}
\pi_{\beta}(t)=\frac{\Gamma\left(\frac{\nu+d}{2}\right)}
{\Gamma\left(\frac{\nu}{2}\right)\left(\nu\pi\right)^{d/2}\mbox{det}\left(\Sigma\right)}
\left[1+\frac{1}{\nu}\left(t-\mu\right)'\Sigma^{-1}\left(t-\mu\right)\right]^{-\left(\frac{\nu+d}{2}\right)},
\label{eq:multivariate_t}
\end{equation}
is sub-exponential. This can be converted to super-exponential by applying the aforementioned 
transformations in succession.

Hence, we can run our geometric ergodic TMCMC algorithms for the super-exponentially light $\pi_{\gamma}$, and
then transform the realizations $\{\gamma^{(k)};k=1,2,\ldots\}$ to 
$\{h(\gamma^{(k)});k=1,2,\ldots\}$. Then it easily follows (see Appendix A of \ctn{Johnson12}) 
that the transformed chain is also geometrically ergodic.

\subsubsection{Simulation studies comparing RWMH and additive TMCMC in the context of diffeomorphism based
simulation from Cauchy and $t$-distributions}
\label{subsubsec:diffeo_simstudy}

We now compare diffeomorphism-based RWMH and additive TMCMC algorithms with respect to K-S distance,
when the target distributions are $d$-dimensional Cauchy and $t$-distributions, the latter having
$\nu$ degrees of freedom. We assume that the location vectors and scale matrices are $\bmu=\bzero_d$ and 
$\bSigma=diag\{0.7\bone_d'\}+0.3\bone_d\bone_d'$, respectively, where 
$\bzero_d$ is a $d$-dimensional vector with all elements $0$, and $\bone_d$ is a $d$-dimensional
vector with each component 1. We choose $d=50$ for the illustrations.
For both RWMH and additive TMCMC we consider the scale of the proposal distribution to be 2.4.

Figure \ref{fig:diffeo_compare_50} compares the performances
of diffeomorphism based RWMH and diffeomorphism based Add-TMCMC with respect to the K-S distance when the 
target distributions are 50-variate Cauchy and 50-variate $t$ respectively,
with the aforementioned location vector and scale matrix. In both the cases Add-TMCMC quite significantly
outperforms RWMH.
Hence, the results are highly encouraging -- additive TMCMC significantly outperforms RWMH 
when the high-dimensional target density is not super-exponential, and is highly dependent. 
Since a mixture of additive and multiplicative TMCMC is demonstrably more efficient than additive TMCMC,
it is clear that the mixture will beat RWMH by a large margin. We have also carried out extensive simulation studies
comparing RWMH and Add-TMCMC when the target distributions are $50$-dimensional $i.i.d.$ Cauchy and $50$-dimensional
$i.i.d.$ $t$ with $10$ degrees of freedom, that is, with $\bmu=\bzero_d$ and $\bSigma=\bI_d$, the latter
standing for the identity matrix or order $d$, with $d=50$. We do not present
the results here due to lack of space, but Add-TMCMC outperformed RWMH at least as significantly
as in this reported dependent set-up.

As an aside, we also compare the gains of the diffeomorphism based approach over the usual, direct application
of RWMH and TMCMC to the target densities. 
Figure \ref{fig:RWMH_w_wo_diffeo_compare_50} compares the performances of diffeomorphism based RWMH and direct RWMH
when the targets are the above-defined 50-dimensional multivariate Cauchy and $t$ (with 10 degrees of freedom).
Likewise, Figure \ref{fig:TMCMC_w_wo_diffeo_compare_50} compares the performances of diffeomorphism based Add-TMCMC
and direct Add-TMCMC with the above 50-dimensional target densities. As is evident from the figures, the
diffeomorphism based approaches quite significantly outperform the direct approaches.

\begin{figure}
\centering
\subfigure [Multivariate Cauchy Target: Diffeomorphism based RWMH vs diffeomorphism based 
Add-TMCMC (scale = 2.4)]{ \label{fig:diffeo1:KS_50}
\includegraphics[width=10cm,height=6cm]
{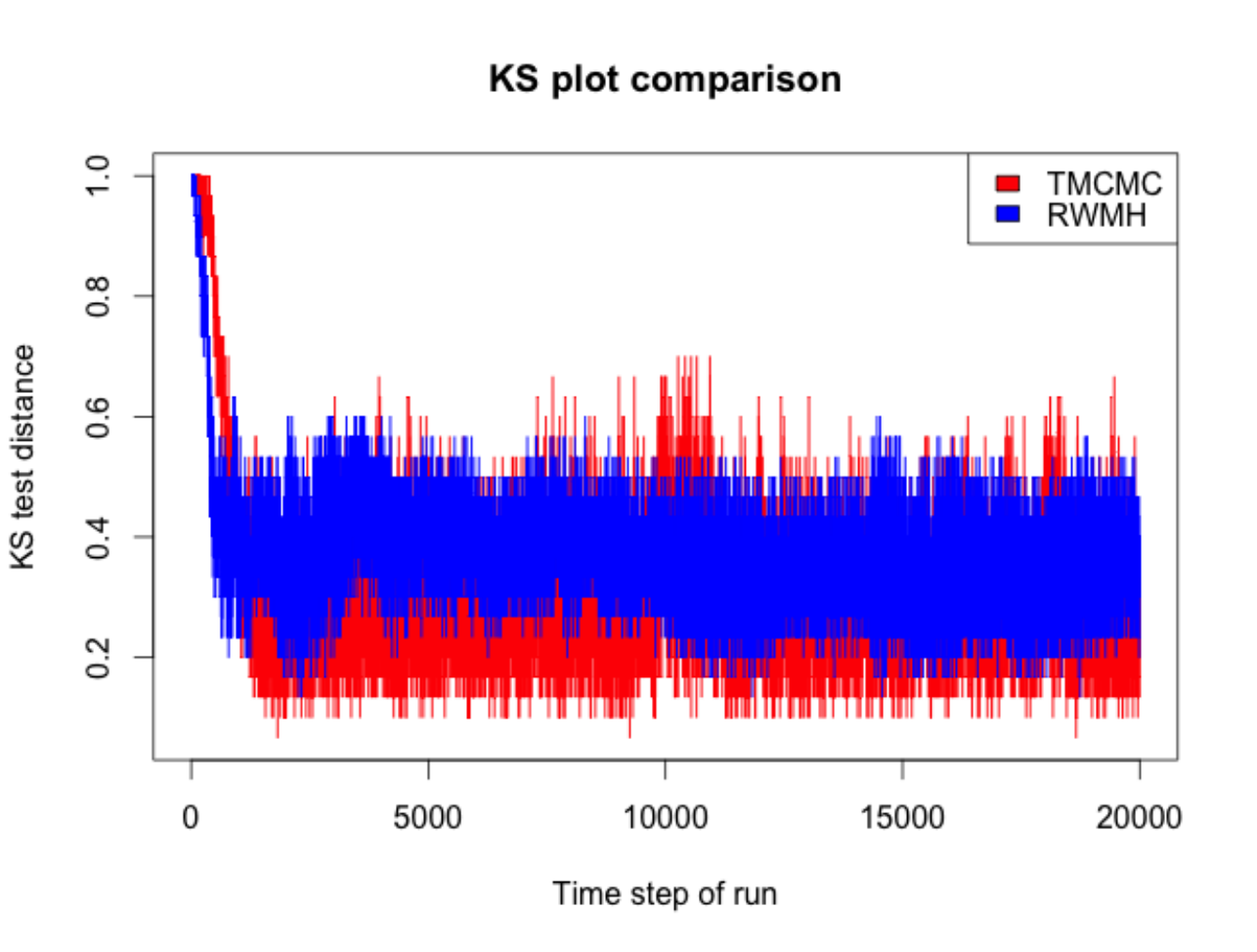}}\\
\subfigure [Multivariate-$t$ Target: Diffeomorphism based RWMH vs diffeomorphism based 
Add-TMCMC (scale = 2.4)]{ \label{fig:diffeo2:KS_50}
\includegraphics[width=10cm,height=6cm]
 {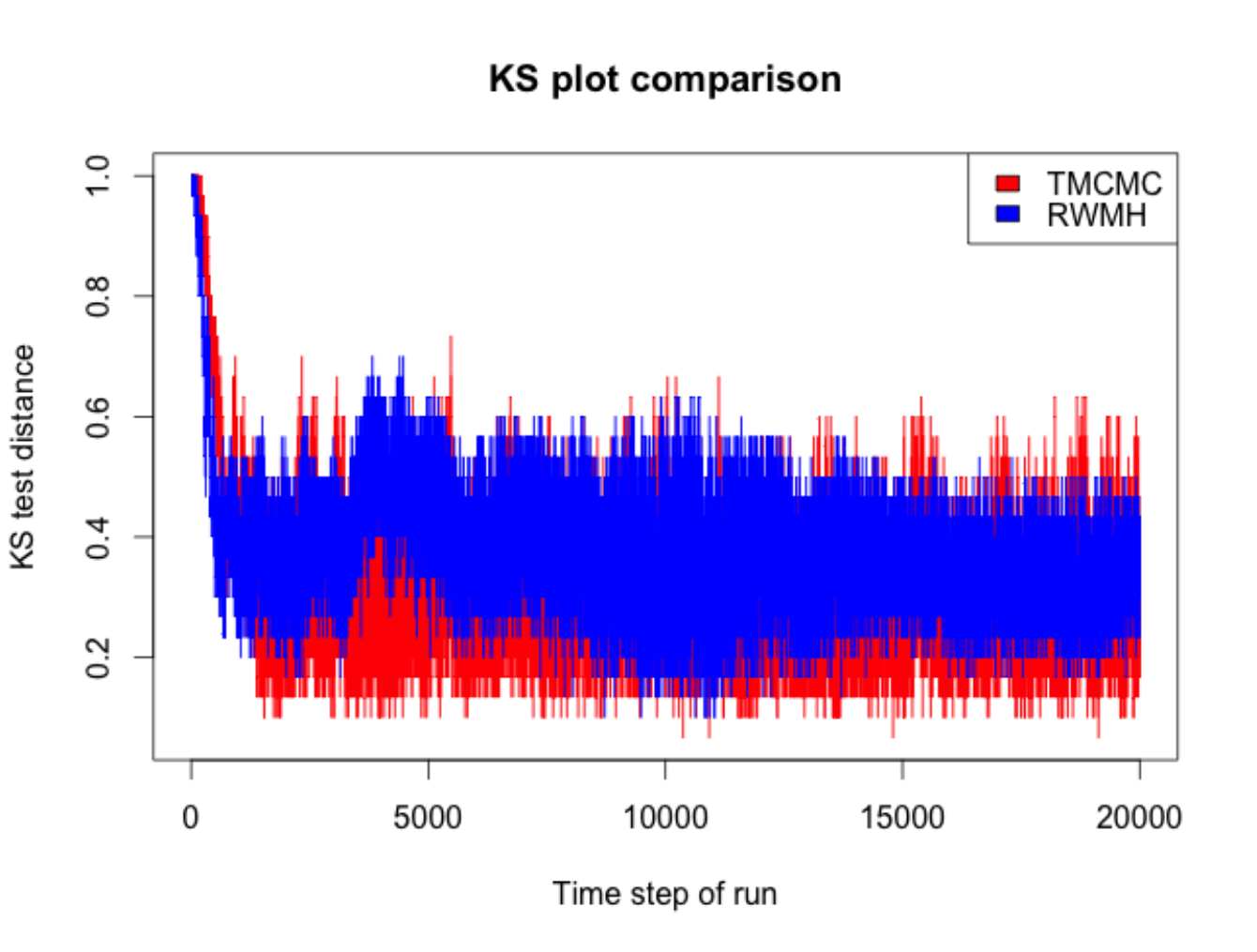}}
\caption{50-dimensional Cauchy and multivariate $t$ (10 degrees of freedom) targets: 
Comparisons between K-S distances associated with diffeomorphism based additive TMCMC and diffeomorphism based RWMH.}
\label{fig:diffeo_compare_50}
\end{figure}

\begin{figure}
\centering
\subfigure [Multivariate Cauchy Target: RWMH with and without diffeomorphism (scale = 2.4)]
{ \label{fig:RWMH_cauchy_w_and_wo_diffeo:KS_50}
\includegraphics[width=10cm,height=6cm]{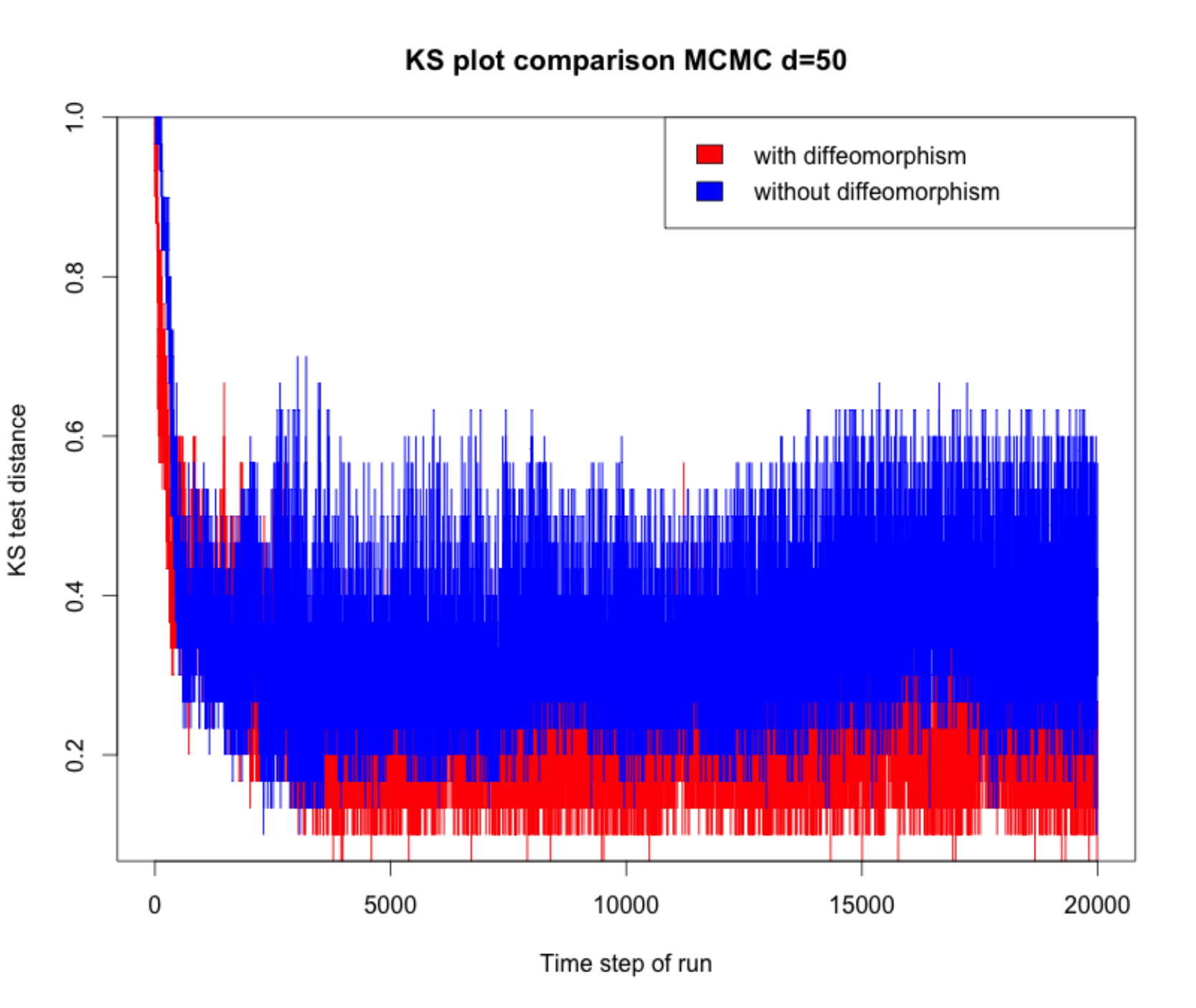}}\\
\subfigure [Multivariate-$t$ Target: RWMH with and without diffeomorphism (scale = 2.4)]
{ \label{fig:RWMH_t_w_and_wo_diffeo:KS_50}
\includegraphics[width=10cm,height=6cm]{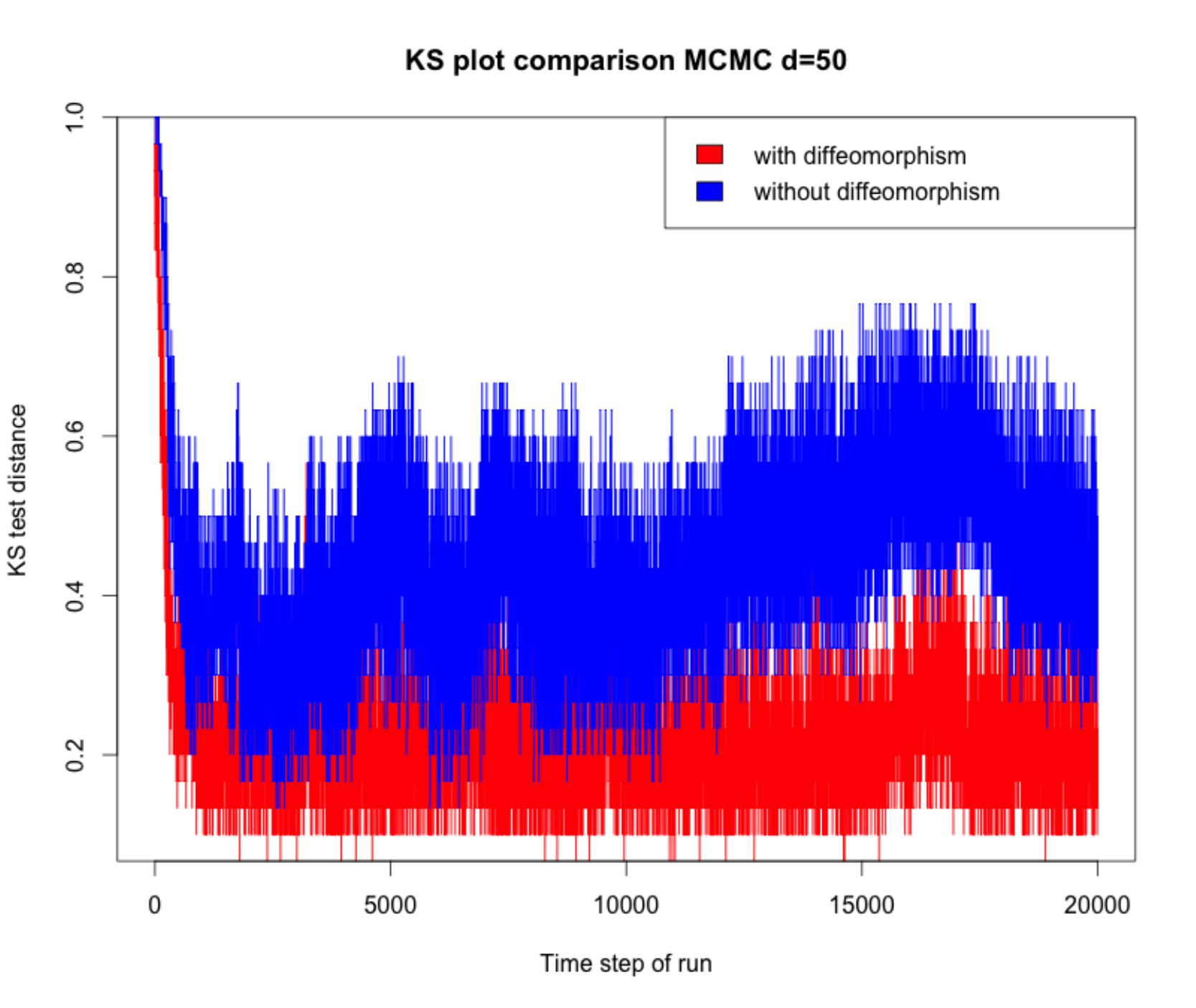}}
\caption{50-dimensional Cauchy and multivariate $t$ (10 degrees of freedom) targets: 
Comparisons between K-S distances associated with RWMH implemented with and without diffeomorphism.}
\label{fig:RWMH_w_wo_diffeo_compare_50}
\end{figure}

\begin{figure}
\centering
\subfigure [Multivariate Cauchy Target: Add-TMCMC with and without diffeomorphism (scale = 2.4)]
{ \label{fig:TMCMC_cauchy_w_and_wo_diffeo:KS_50}
\includegraphics[width=10cm,height=6cm]{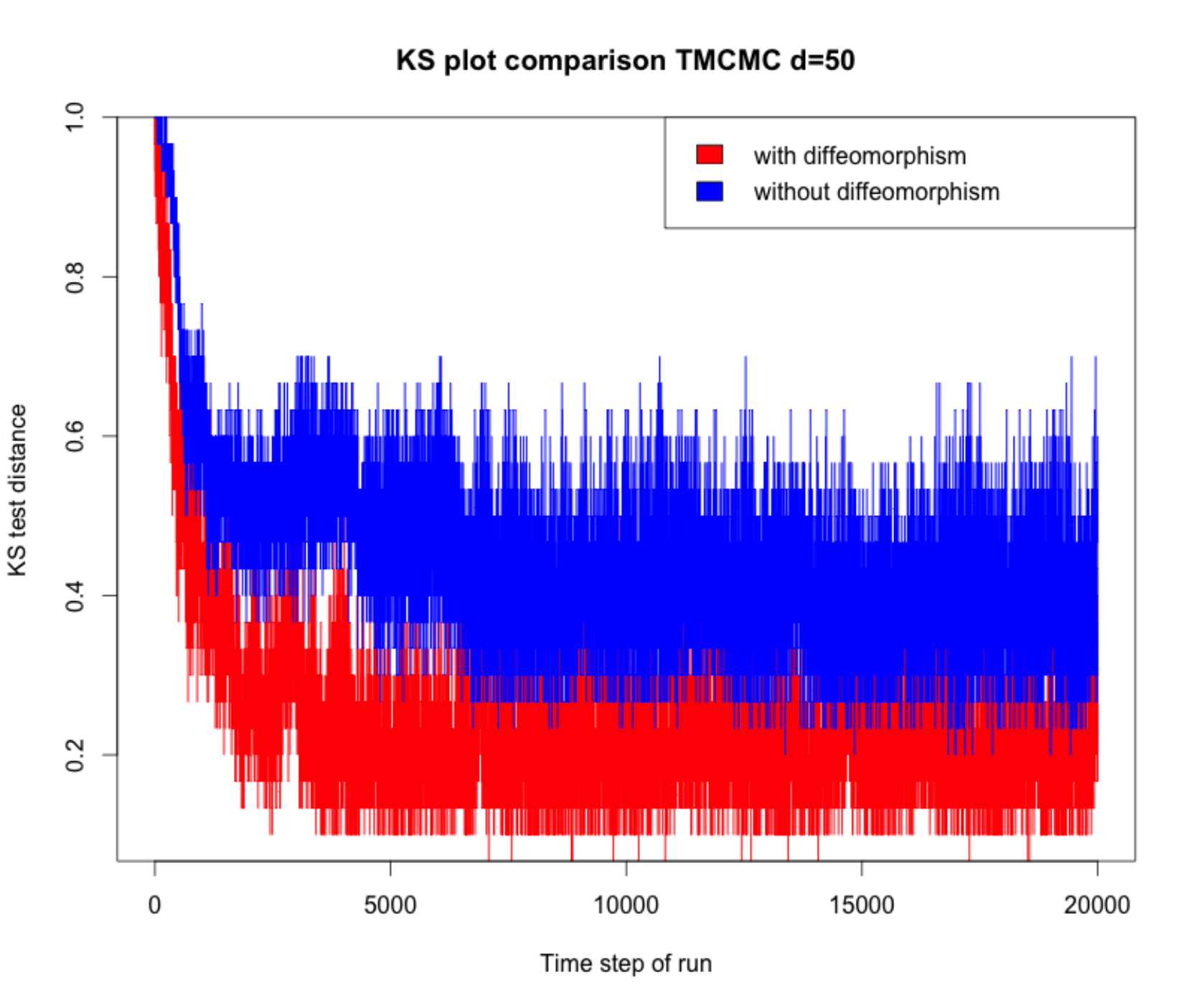}}\\
\subfigure [Multivariate-$t$ Target: Add-TMCMC with and without diffeomorphism (scale = 2.4)]
{ \label{fig:TMCMC_t_w_and_wo_diffeo:KS_50}
\includegraphics[width=10cm,height=6cm]{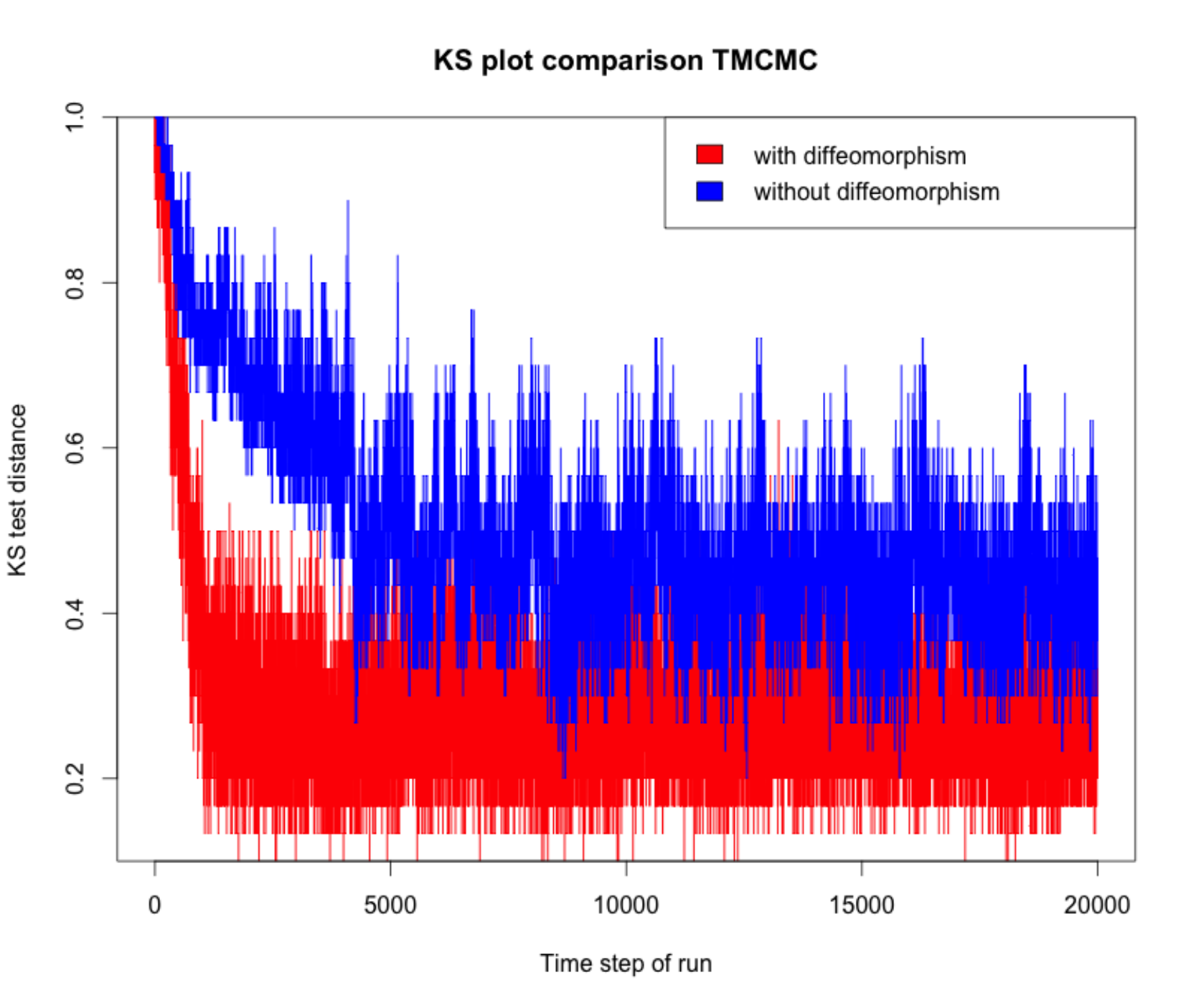}}
\caption{50-dimensional Cauchy and multivariate $t$ (10 degrees of freedom) targets: 
Comparisons between K-S distances associated with Add-TMCMC implemented with and without diffeomorphism.}
\label{fig:TMCMC_w_wo_diffeo_compare_50}
\end{figure}

\section{Concluding remarks}
\label{sec:conclusions}

We presented a comprehensive comparative study of geometric ergodicity and convergence behavior of various versions of TMCMC:
additive, ``essentially full" multiplicative and mixture TMCMC. Additive TMCMC is the easiest to implement and as observed in the
simulation study, has somewhat better convergence to the target distribution compared to RWMH. The essentially 
fully multiplicative 
TMCMC traverses the sample space more rapidly but we observed that it is relatively slow in convergence 
to the target density  
compared to the standard RWMH approach. The best convergence results are obtained for mixture TMCMC 
which combines the 
additive and the multiplicative moves in equal proportions.

Of considerable interest are situations when the high-dimensional target densities are not super-exponential
but can be handled by the diffeomorphism based approach. The relevant simulation studies detailed in 
Section \ref{subsubsec:diffeo_simstudy} demonstrate far superior convergence of additive TMCMC compared
to RWMH. Since these simulation studies are conducted assuming high dependence structure of the target densities,
the results are particularly encouraging and lead us to recommend TMCMC in general situations. Moreover,
it is to be noted that in these simulation studies we concern ourselves with only additive TMCMC. Since
a mixture of additive and multiplicative TMCMC is seen to be more efficient in comparison with additive TMCMC,
it is clear that such a mixture will outperform RWMH by even greater margins.

There are obviously some questions of further interest. We would definitely like to have quantitative rates of convergence for 
each of the three approaches to TMCMC. In this paper we considered the mixing proportion in mixture TMCMC to be $1/2$ 
and we also observed in our simulation study that extremal mixing proportions (which correspond to additive and 
essentially fully 
multiplicative approaches) lead to slower convergence compared to uniform mixing. But it would be worth noting how this rate 
of convergence changes with the change in mixing proportion. 

Optimal scaling of TMCMC methods is another area which is of considerable interest to us.
The optimal scaling for additive TMCMC has been studied for a broad 
class of multivariate target densities (\ctn{Dey13}), but the optimal scaling for mixture TMCMC and 
multiplicative or essentially fully 
multiplicative approaches are yet to be determined. The biggest challenge in dealing with this problem is that the generator functions 
for the associated time scaled diffusion process for these methods are hard to express in any simple analytic form. 

One area we are currently focussing on is defining adaptive versions of the TMCMC approach (additive and multiplicative) 
and comparing the performances (convergence criterion and acceptance rate in particular) among various adaptive schemes 
and also with the typical non adaptive algorithms we considered here. 

We are also trying to expand the scope of our approach beyond 
$\mathbb{R}^{d}$ by considering spheres and other Riemannian or Symplectic manifolds as the support of the target distributions and 
it would be interesting to investigate such properties like irreducibility, detailed balance and 
ergodicity properties of the TMCMC algorithms over such spaces.

\section*{Acknowledgment}
We are sincerely grateful to three anonymous reviewers 
whose comments led to a much improved version of our manuscript.

\section*{Appendix}
\begin{appendix}

\section{Minorization condition for multiplicative TMCMC}
\label{sec:minorization}

For the one-dimensional case, minorization conditions of multiplicative TMCMC has been established
by \ctn{Dutta12}. Here we generalize the results to arbitrary dimension.
For simplicity we assume
$p_i=q_i=1/3$ for $i=1,\ldots,d$. 
The following theorem establishes the minorization condition for multiplicative TMCMC.

\begin{theorem}
\label{theorem:minor_mult}
Let the target density $\pi$ be bounded and positive on compact sets. Then there exists a nonzero
measure $\nu$, a positive integer $m$, $\delta>0$, and a {\it small set} $E^*$ such that
\begin{equation}
\left\{P^{(2)}\right\}^m(x,\mathbb A)\geq\delta \nu(\mathbb A),\quad\forall x\in E^*\quad\mbox{and for all Borel sets}\ \ \mathbb A. 
\label{eq:minor1}
\end{equation}
\end{theorem}
\begin{proof}
Observe that, from $x=(x_1,\ldots,x_d)$ it is possible to move to any Borel set $\mathbb A$ 
in at least $d$ steps using  
those multiplicative TMCMC move types $b=(b_1,\ldots,b_d)$ which update only one coordinate at a time.
Hence, for our purpose it is sufficient to confine attention to these moves.

Let $E^*$ denote a compact subset of $\mathbb R^d$. Also, let $\mathbb C$ be a compact set containing $E^*$.
Let $\mathbb A^*=\mathbb A\cap \mathbb C$. For the simplicity of presentation we present the proof
of minorization for $d=2$.

Let $\mathbb A_1=\{(\epsilon_1,\epsilon_2):(x_1\epsilon_1,x_2\epsilon_2)\in\mathbb A^*\}$, 
$\mathbb A_2=\{(\epsilon_1,\epsilon_2):(x_1/\epsilon_1,x_2/\epsilon_2)\in\mathbb A^*\}$,
$\mathbb A_3=\{(\epsilon_1,\epsilon_2):(x_1\epsilon_1,x_2/\epsilon_2)\in\mathbb A^*\}$, and
$\mathbb A_4=\{(\epsilon_1,\epsilon_2):(x_1/\epsilon_1,x_2\epsilon_2)\in\mathbb A^*\}$.

For $x\in E^*$, we have
\begin{align}
&\left\{P^{(2)}\right\}^2(x,\mathbb A)\geq \left\{P^{(2)}\right\}^2(x,\mathbb A^*)\notag\\  
&\geq \frac{1}{3^4}
\int_{\mathbb A_1}
\min\left\{1,\frac{\pi(x_1\epsilon_1,x_2)|\epsilon_1|}{\pi(x_1,x_2)}\right\}
\times \min\left\{1,\frac{\pi(x_1\epsilon_1,x_2\epsilon_2)|\epsilon_2|}{\pi(x_1\epsilon_1,x_2)}\right\}
g^{(2)}(\epsilon_1)g^{(2)}(\epsilon_2)d\epsilon_1d\epsilon_2\notag\\
&+ \frac{1}{3^4}
\int_{\mathbb A_2}
\min\left\{1,\frac{\pi(x_1/\epsilon_1,x_2)|\epsilon|^{-1}}{\pi(x_1,x_2)}\right\}
\times \min\left\{1,\frac{\pi(x_1/\epsilon_1,x_2/\epsilon_2)|\epsilon_2|^{-1}}{\pi(x_1/\epsilon_1,x_2)}\right\}
g^{(2)}(\epsilon_1)g^{(2)}(\epsilon_2)d\epsilon_1d\epsilon_2\notag\\
&+ \frac{1}{3^4}
\int_{\mathbb A_3}
\min\left\{1,\frac{\pi(x_1\epsilon_1,x_2)|\epsilon_1|}{\pi(x_1,x_2)}\right\}
\times \min\left\{1,\frac{\pi(x_1\epsilon_1,x_2/\epsilon_2)|\epsilon_2|^{-1}}{\pi(x_1\epsilon_1,x_2)}\right\}
g^{(2)}(\epsilon_1)g^{(2)}(\epsilon_2)d\epsilon_1d\epsilon_2\notag\\
&+ \frac{1}{3^4}
\int_{\mathbb A_4}
\min\left\{1,\frac{\pi(x_1/\epsilon_1,x_2)|\epsilon_1|^{-1}}{\pi(x_1,x_2)}\right\}
\times \min\left\{1,\frac{\pi(x_1/\epsilon_1,x_2\epsilon_2)|\epsilon_2|}{\pi(x_1/\epsilon_1,x_2)}\right\}
g^{(2)}(\epsilon_1)g^{(2)}(\epsilon_2)d\epsilon_1d\epsilon_2.
\label{eq:minorization}
\end{align}

Let $r=\inf_{y\in\mathbb C}\pi(y)$ and $R=\sup_{y\in\mathbb C}\pi(y)$.
Also note that each integral on $\mathbb A_i$; $i=1,2,3,4$, can be split into 
$\mathbb A_i=\{\mathbb A_i\cap\mathbb S_{\eta}\}\cup\{\mathbb A_i\cap\mathbb S^c_{\eta}\}$, where 
$\mathbb S_{\eta}=\{(\epsilon_1,\epsilon_2):\eta<|\epsilon_1|\leq 1, \eta<|\epsilon_2|\leq 1\}$, for some $\eta>0$,
and $\mathbb S^c_{\eta}$ denotes the complement of $\mathbb S_{\eta}$. 
Let $G$ denote the probability measure corresponding
to the distribution $\epsilon_1,\epsilon_2\stackrel{i.i.d.}{\sim}g^{(2)}$. 

On $\mathbb A_i\cap\mathbb S^c_{\eta}$, for $i=1,3,4$, the corresponding integrands
have infimum zero; hence zero is the lower bound of the respective integrals on 
$\mathbb A_i\cap\mathbb S^c_{\eta}$, for $i=1,3,4$. 
On $\mathbb A_2\cap\mathbb S^c_{\eta}$, the 
integrand of the second integral has infimum equal to 1; hence, 
the corresponding integral is bounded below by
$\frac{1}{3^4}G(\mathbb A_2\cap\mathbb S^c_{\eta})$.
Note that $G(\mathbb A_2\cap\mathbb S^c_{\eta})$ can be made
arbitrarily small by choosing $\eta$ to be as small as desired.

On $\mathbb A_i\cap\mathbb S_{\eta}$, each of the integrals are bounded below by 
$\frac{\eta^2}{3^4}(\frac{r}{R})^2G(\mathbb A_i\cap\mathbb S_{\eta})$.
Hence,
\begin{align}
\left\{P^{(2)}\right\}^2(x,\mathbb A)&\geq \left\{P^{(2)}\right\}^2(x,\mathbb A^*)\notag\\  
&\geq \frac{\eta^2}{3^4}\left(\frac{r}{R}\right)^2\sum_{i=1}^4G(\mathbb A_i\cap\mathbb S_{\eta})\notag\\
&\geq \frac{\eta^2}{3^4}\left(\frac{r}{R}\right)^2G(\left\{\cup_{i=1}^4\mathbb A_i\right\}\cap\mathbb S_{\eta})\notag\\
&= \frac{\eta^2}{3^4}\left(\frac{r}{R}\right)^2G(\mathbb A^*\cap\mathbb S_{\eta}).\notag\\
&= \frac{\eta^2}{3^4}\left(\frac{r}{R}\right)^2G(\mathbb S_{\eta})
\times\frac{G(\mathbb A^*\cap\mathbb S_{\eta})}{G(\mathbb S_{\eta})}.\notag\\
&=\delta\nu(\mathbb A^*),
\end{align}
with 
\[
\delta=\frac{\eta^2}{3^4}\left(\frac{r}{R}\right)^2G(\mathbb S_{\eta})\quad\mbox{and}\quad 
\nu(\mathbb A^*)=\frac{G(\mathbb A^*\cap\mathbb S_{\eta})}{G(\mathbb S_{\eta})}.
\]
Hence, minorization holds for multiplicative TMCMC, and $E^*$ is the small set. 
The same ideas of the proof go through for any finite dimension $d$.
\end{proof}

We next show that vectors in the set 
\[
\mathcal V=\{(v_1,\ldots,v_d)\in\mathbb R^d: v_i=0~ \mbox{for at least one}~i\in\{1,\ldots,d\}\},
\]
can not be limit points of small sets. 
For our purpose we need a lemma
which can be seen as a generalization of Lemma 1 of \ctn{Dutta12} to arbitrary dimensions
and for vectors in $\mathcal V$.
\begin{lemma}
\label{lemma:minor_mult}
Fix $v=(v_1,\ldots,v_d)\in\mathcal V$. For $\{i_1,\ldots,i_k\}\subseteq\{1,\ldots,d\}$, where $k\leq d$, let
$v_{i_j}=0$, for $j=1,\ldots,k$.
Let $\{x_n\}$ be a sequence of positive (negative) numbers decreasing (increasing) to zero. 
Consider the sequence $v_n=(v_{1,n},\ldots,v_{d,n})'$, where $v_{j,n}=x_n$ for $j=i_1,\ldots,i_k$,
and $v_{j,n}=v_j$ for $j\in\{1,\ldots,d\}\backslash \{i_1,\ldots,i_k\}$. If $v_i=0$ for $i=1,\ldots,d$,
then $v_n=(x_n,\ldots,x_n)'$ may also be considered.
Then,
\begin{equation}
P^{(2)}(v_n,\mathbb A)\rightarrow 0, 
\label{eq:minor_mult}
\end{equation}
for all Borel sets $\mathbb A$ such that $\mathbb A\cap\{(v_1,\ldots,v_d)\in\mathbb R^d:v_{i_j}=0;~j=1,\ldots,k\}=\emptyset$.
\end{lemma}
\begin{proof}
Without loss of generality we present the proof for $d=2$. Let us fix $v=(v_1,v_2)$, where $v_1=0$ and $v_2\in\mathbb R$.
Let $v_n=(x_n,v_2)$. Note that for moving from $x_n$ to $z\in\mathbb R$, where $|x_n|\leq |z|$ for all $n$, we must simulate
$\epsilon = x_n/z$ and take the backward move $z=x_n/\epsilon$. The move $z=x_n\epsilon$, with $\epsilon=z/x_n$
can not be valid in this case, since $x_n\rightarrow 0$ implies that for large $n$, $\epsilon\notin [-1,1]$.

Since the acceptance probability is bounded above by 1, we have, for $y<0$,
\begin{align}
P^{(2)}(v_n,(-\infty,y]\times (-\infty,\infty))
&\leq \frac{1}{3^2}\int_{x_n/y}^0g(\epsilon)d\epsilon\notag\\
&\rightarrow 0.
\label{eq:limit_point1}
\end{align}
If $y>0$, then
\begin{align}
P^{(2)}(v_n,[y,\infty)\times (-\infty,\infty))
&\leq \frac{1}{3^2}\int_0^{x_n/y}g(\epsilon)d\epsilon\notag\\
&\rightarrow 0.
\label{eq:limit_point2}
\end{align}
Hence, (\ref{eq:minor_mult}) holds when $d=2$. The proof clearly goes through for any dimension $d$. 

If $v=(0,0)$, we can consider $v_n=(x_n,0)'$ or $v_n=(x_n,x_n)'$. Then, in addition to
(\ref{eq:limit_point1}) and (\ref{eq:limit_point2}), which clearly hold, the following also hold true:
if $y<0$
\[
P^{(2)}(v_n,(-\infty,\infty)\times (-\infty,y])\rightarrow 0,
\]
and
\[
P^{(2)}(v_n,(-\infty,\infty)\times [y,\infty))\rightarrow 0,
\]
if $y>0$.
These imply that for dimension $d=2$,
\begin{equation}
P^{(2)}(v_n,\cdot)\rightarrow I_{\{\bzero\}}(\cdot).
\label{eq:limit_point_0}
\end{equation}
The above result (\ref{eq:limit_point_0}) clearly holds for any dimension $d$ for
$v=(0,0,\ldots,0)'$ and $v_n=x_n\bone$, where $\bone=(1,1,\ldots,1)'$ is the $d$-component
vector of ones.
\end{proof}

Now, if $v\in\mathcal V$ is a limit point of $E^*$, then there exists a sequence $v_n$ as in Lemma \ref{lemma:minor_mult},
converging to $v$.  
This, and Lemma \ref{lemma:minor_mult} imply that for any fixed integer $m>1$, and for any Borel set $\mathbb A$, 
\begin{align}
\left\{P^{(2)}\right\}^m(v_n,\mathbb A)
&=\int_{\mathbb R^d}\left\{P^{(2)}\right\}^{m-1}(z,\mathbb A)P^{(2)}(v_n,dz)\notag\\
&\rightarrow 0,
\label{eq:prob_lim1}
\end{align}
if $\mathbb A\cap\{(v_1,\ldots,v_d)\in\mathbb R^d:v_{i_j}=0;~j=1,\ldots,k\}=\emptyset$.

In particular, if $\bzero$ is a limit point of $E$, then 
for any fixed integer $m>1$, and for any Borel set $\mathbb A$, 
\begin{align}
\left\{P^{(2)}\right\}^m(x_n\bone,\mathbb A)
&\rightarrow I_{\{\bzero\}}(\mathbb A).
\label{eq:prob_lim2}
\end{align}
Both (\ref{eq:prob_lim1}) and (\ref{eq:prob_lim2})
contradict the minorization inequality (\ref{eq:minor1}).

Now consider the case of additive-multiplicative TMCMC. Let the coordinates with indices 
$\{j_1,j_2,\ldots,j_{\ell}\}\subset\{1,2,\ldots,d\}$
be given the multiplicative transformation and let the remaining coordinates be given the additive transformation.
Here, let $\mathcal V(j_1,\ldots,j_{\ell})=
\{(v_1,\ldots,v_d)'\in\mathbb R^d:v_j=0~\mbox{for at least one}~j\in\{j_1,j_2,\ldots,j_{\ell}\}\}$. Then
vectors $v\in \mathcal V(j_1,\ldots,j_{\ell})$ can not be limit points of small sets associated 
with additive-multiplicative TMCMC.
In particular, $\bzero$ can not be a limit point.
The proof is the same as in the case of multiplicative TMCMC, and hence omitted.

\section{Proof of geometric ergodicity of the Markov transition kernel 
$P=\pi(\mathbb N_0)P^{(1)}+\pi(\mathbb N^c_0)P^{(2)}$}
\label{sec:P_geo}

Let us first introduce an auxiliary random variable $Z$, with
\begin{align}
Pr(Z=1)&=\pi(\mathbb N_0)\quad\mbox{and}\quad  Pr(Z=2)=1-Pr(Z=1).\label{eq:pi_Z}
\end{align}
Note that for $i=1,2$,
\begin{equation}
P(x,A|Z=i)=P^{(i)}(x,A)\quad\mbox{and}\quad \pi(A|Z=i)=\pi_i.
\label{eq:pi_i}
\end{equation}
Also note that, since $P^{(i)}$ is geometrically ergodic when the target density
is $\pi_i$, we must have
\begin{equation}
\left\|\left\{P^{(i)}\right\}^n(x,\cdot)-\pi_i(\cdot)\right\|\leq M_i(x)\rho^n_i,
\label{eq:pi_geo}
\end{equation}
for $i=1,2$, for some $M_1(x),M_2(x)<\infty$ and $0<\rho_1,\rho_2<1$.

Now,
\begin{align}
&\|P^n(x,\cdot)-\pi(\cdot)\|_{TV}=\underset{A\in\mathcal B(\mathbb R^d)}
\sup\bigg |P^n(x,A)-\pi(A)\bigg |\notag\\
&= \underset{A\in\mathcal B(\mathbb R^d)}
\sup\bigg |P^n(x,A|Z=1)Pr(Z=1)+ P^n(x,A|Z=2)Pr(Z=2)\notag\\
&\quad\quad-\left(\pi(A|Z=1)Pr(Z=1)+\pi(A|Z=2)Pr(Z=2)\right)\bigg |\notag\\
&= \underset{A\in\mathcal B(\mathbb R^d)}
\sup\bigg |\left\{P^{(1)}\right\}^n(x,A)Pr(Z=1)+ \left\{P^{(2)}\right\}^n(x,A)Pr(Z=2)\notag\\
&\quad\quad-\left(\pi_1(A)Pr(Z=1)+\pi_2(A)Pr(Z=2)\right)\bigg |\quad\mbox{by (\ref{eq:pi_i})}\notag\\
&\leq Pr(Z=1)\left\|\left\{P^{(1)}\right\}^n(x,\cdot)-\pi_1(\cdot)\right\|
+Pr(Z=2)\left\|\left\{P^{(2)}\right\}^n(x,\cdot)-\pi_2(\cdot)\right\|\notag\\
&\leq Pr(Z=1)M_1(x)\rho^n_1+Pr(Z=2)M_2(x)\rho^n_2\quad\mbox{by (\ref{eq:pi_geo})}\notag\\
&\leq M(x)\rho^n,\notag
\end{align}
where $M(x)\geq\max\{M_1(x),M_2(x)\}$, and $\rho\geq\max\{\rho_1,\rho_2\}$.
Hence, $P$ is geometrically ergodic when the target density is $\pi$.

Note that the proof employed in Section \ref{subsec:usual_mixture_kernel} for showing 
geometric ergodicity of the alternative mixture Markov transition kernel $P^*$, is also valid
for showing geometric ergodicity of $P$, but the current proof (with slight modification;
replacing the summations with integrations) is appropriate for
proving geometric ergodicity of continuous mixture kernels of the form 
(\ref{eq:non_super_exp_P}) for continuous mixture target densities of the form
(\ref{eq:non_super_exp_pi}) since a single function $V$ need not be appropriate 
for (uncountably) infinite number of mixture components.

\section{Discussion on estimation of the mixing probability $\pi(\mathbb N_0)$}
\label{sec:implement_P}
In order to implement the Markov transition kernel $P$, for each $k=1,2,\ldots$, 
we are required to draw $u\sim U(0,1)$;
if $u<\pi(\mathbb N_0)$, we select $x^{(k)}_1$, else we select $x^{(k)}_2$. Note that 
$\pi(\mathbb N_0)$ is not known, and needs to be estimated numerically.
Direct estimation using TMCMC samples from $\pi$ will generally not be reliable, since
the region $\mathbb N_0$, being arbitrarily small, can be easily missed by any MCMC method.
However, this may be reliably estimated using importance sampling as follows.

Let $\pi(x)=c\ell(x)$, where $c=1/\int\ell(y)dy$ is the unknown normalizing constant.
Also, let $h(x)=|\mathbb N_0|^{-1}I_{\mathbb N_0}(x)$ be the uniform distribution
on $\mathbb N_0$, where $|\mathbb N_0|$ denotes the Lebesgue measure of the set $\mathbb N_0$.
We may use $h$ as the importance sampling density in the region $\mathbb N_0$. For the region $\mathbb N^c_0$
we may consider some thick-tailed importance sampling density $g(x)$, for example, a $d$-variate 
$t$-density, but adjusting the support to be $\mathbb N^c_0$.
Then
\begin{align}
\pi(\mathbb N_0) &=\frac{\int_{\mathbb N_0}\ell(x)dx}{\int\ell(x)dx}
=\frac{\int_{\mathbb N_0}\frac{\ell(x)}{h(x)}h(x)dx}
{\int_{\mathbb N_0}\frac{\ell(x)}{h(x)}h(x)dx+\int_{\mathbb N^c_0}\frac{\ell(x)}{g(x)}g(x)dx}\notag\\
&\approx\frac{\frac{1}{N_1}\sum_{j=1}^{N_1}\frac{\ell(x^{(j)})}{h(x^{(j)})}}
{\frac{1}{N_1}\sum_{j=1}^{N_1}\frac{\ell(x^{(j)})}{h(x^{(j)})}
+\frac{1}{N_2}\sum_{k=1}^{N_2}\frac{\ell(y^{(k)})}{g(y^{(k)})}}=\hat\pi(\mathbb N_0)~\mbox{(say)},\notag
\end{align}
where $\{x^{(j)};j=1,\ldots,N_1\}$ are $i.i.d.$ realizations drawn from the uniform distribution $h$ 
and $\{y^{(k)};k=1,\ldots,N_2\}$ are $i.i.d.$ or TMCMC realizations from $g$, depending on the complexity
of the form of $g$. 
The parameters of $g$ may be chosen by variational methods; see http://www.gatsby.ucl.ac.uk/vbayes/
for a vast repository of papers, softwares and links on variational methods.

Observe that even though we are proposing to estimate $\pi(\mathbb N_0)$ by $\hat\pi(\mathbb N_0)$, 
implementation of the mixture kernel $P$ with $\hat\pi(\mathbb N_0)$ as the mixing probability
is expected to be exactly the same as the mixture kernel $P$ with the true mixing probability $\pi(\mathbb N_0)$.
This is because even if $\hat\pi(\mathbb N_0)$ is only a reasonably accurate estimate of 
$\pi(\mathbb N_0)$, it is expected that
for any $u\sim U(0,1)$, $u<\pi(\mathbb N_0)$ if and only if $u<\hat\pi(\mathbb N_0)$.
For instance, if $\hat\pi(\mathbb N_0)=\pi(\mathbb N_0)+\eta$,
for some $\eta>0$, then $Pr\left(\pi(\mathbb N_0)<u<\hat\pi(\mathbb N_0)\right)=\eta$.
Even if $\eta$ is not extremely small, the above probability is still reasonably small,
for reasonably small values of $\eta$.
In other words, a very high degree of accuracy of the estimate $\hat\pi(\mathbb N_0)$ 
is not that important in this case.

\end{appendix}

\bibliographystyle{natbib}
\bibliography{irmcmc}

\begin{thebibliography}{}

\bibitem[Dey and Bhattacharya(2016)Dey and Bhattacharya]{Dey13}
Dey, K.~K. and Bhattacharya, S. (2016).
\newblock {A} {B}rief {T}utorial on {T}ransformation {B}ased {M}arkov {C}hain
  {M}onte {C}arlo and {O}ptimal {S}caling of the {A}dditive {T}ransformation.
\newblock To appear in {\it Brazilian Journal of Probability and Statistics}.
  Available at http://arxiv.org/abs/1307.1446.

\bibitem[Dutta(2012)Dutta]{Dutta12}
Dutta, S. (2012).
\newblock {M}ultiplicative {R}andom {W}alk {M}etropolis-{H}astings on the
  {R}eal {L}ine.
\newblock {\em Sankhya B\/}, {\bf 74}, 315--342.

\bibitem[Dutta and Bhattacharya(2014)Dutta and Bhattacharya]{Dutta13}
Dutta, S. and Bhattacharya, S. (2014).
\newblock {M}arkov {C}hain {M}onte {C}arlo {B}ased on {D}eterministic
  {T}ransformations.
\newblock {\em Statistical Methodology\/}, {\bf 16}, 100--116.
\newblock Also available at http://arxiv.org/abs/1106.5850. Supplement
  available at http://arxiv.org/abs/1306.6684.

\bibitem[Jarner and Hansen(2000)Jarner and Hansen]{Jarner00}
Jarner, S.~F. and Hansen, E. (2000).
\newblock {G}eometric {E}rgodicity of {M}etropolis {A}lgorithms.
\newblock {\em Stochastic Processes and their Applications\/}, {\bf 85},
  341--361.

\bibitem[Jarner and Roberts(2002)Jarner and Roberts]{Jarner02}
Jarner, S.~F. and Roberts, G.~O. (2002).
\newblock {P}olynomial {C}onvergence {R}ates of {M}arkov {C}hains.
\newblock {\em Annals of Applied Probability\/}, {\bf 12}, 224--247.

\bibitem[Jarner and Roberts(2007)Jarner and Roberts]{Jarner07}
Jarner, S.~F. and Roberts, G.~O. (2007).
\newblock {C}onvergence of {H}eavy-{T}ailed {M}onte {C}arlo {M}arkov {C}hain
  {A}lgorithms.
\newblock {\em Scandinavian Journal of Statistics\/}, {\bf 34}, 781--815.

\bibitem[Johnson and Geyer(2012)Johnson and Geyer]{Johnson12}
Johnson, L.~T. and Geyer (2012).
\newblock {V}ariable {T}ransformation to {O}btain {G}eometric {E}rgodicity in
  the {R}andom-{W}alk {M}etropolis {A}lgorithm.
\newblock {\em The Annals of Statistics\/}, {\bf 40}, 3050--3076.

\bibitem[Jones and Hobert(2001)Jones and Hobert]{jones01}
Jones, G.~L. and Hobert, J.~P. (2001).
\newblock {H}onest {E}xploration of {I}ntractable {P}robability {D}istributions
  via {M}arkov {C}hain {M}onte {C}arlo.
\newblock {\em Statistical Science\/}, {\bf 16}(4), 312--334.

\bibitem[Mengersen and Tweedie(1996)Mengersen and Tweedie]{Mengersen96}
Mengersen, K.~L. and Tweedie, R.~L. (1996).
\newblock {R}ates of {C}onvergence of the {H}astings and {M}etropolis
  {A}lgorithms.
\newblock {\em The Annals of Statistics\/}, {\bf 24}, 101--121.

\bibitem[Meyn and Tweedie(1993)Meyn and Tweedie]{Meyn93}
Meyn, S.~P. and Tweedie, R.~L. (1993).
\newblock {\em {M}arkov {C}hains and {S}tochastic {S}tability\/}.
\newblock Springer-Verlag, London.

\bibitem[Roberts and Tweedie(1996)Roberts and Tweedie]{Roberts96}
Roberts, G.~O. and Tweedie, R.~L. (1996).
\newblock {G}eometric {C}onvergence and {C}entral {L}imit {T}heorems for
  {M}ultidimensional {H}astings and {M}etropolis {A}lgorithms.
\newblock {\em Biometrika\/}, {\bf 83}, 95--110.

\bibitem[Roberts {\em et~al.}(1997)Roberts, Gelman, and Gilks]{Roberts1997}
Roberts, G.~O., Gelman, A., and Gilks, W.~R. (1997).
\newblock {W}eak {C}onvergence and {O}ptimal {S}caling of {R}andom {W}alk
  {M}etropolis {A}lgorithms.
\newblock {\em The Annals of Applied Probability\/}, {\bf 7}, 110--120.

\end{thebibliography}

\end{document}